\documentclass[11pt,final]{article}
\usepackage{amssymb,paralist,amsthm}
\usepackage[raiselinks=false,colorlinks=true,citecolor=blue,
    urlcolor=blue,linkcolor=blue,bookmarksopen=true]{hyperref}
\usepackage{calc,bm}
\usepackage{bookmark}
\usepackage[usenames,dvipsnames]{xcolor}
\usepackage[frame,matrix,dvips,color,arrow,curve,poly,line]{xy} 
\UseCrayolaColors
\usepackage{fullpage}
\newtheorem{theorem}{Theorem}
\newtheorem{lemma}[theorem]{Lemma} 
\newtheorem{corollary}[theorem]{Corollary}
\theoremstyle{definition}
\newtheorem{definition}{Definition}

\makeatletter
 \newcommand\setreflabel[1]{\protected@edef\@currentlabel{#1}}
 \newcounter{claim}[theorem]
 
 \newcommand\claimlabelfmt[1]{(#1)}
 \newenvironment{claim}[1][]%
   {\removelastskip\ifx\newenvironment#1\newenvironment%
    \refstepcounter{claim}\def\@claim{{\arabic{theorem}.\arabic{claim}}}\else\def\@claim{#1}\fi%
    \setreflabel{\expandafter\claimlabelfmt{\@claim}}%
    \global\edef\@lastclaim{clm@\the\inputlineno}%
    \label{\@lastclaim}\par\vspace{1ex}%
    \begin{compactitem}[~{\@currentlabel}]\item\it}
   {\end{compactitem}\par\vspace{1ex}}
 \newenvironment{proofclaim}[1][]%
   {\par\edef\@thisclaim{\ifx\newenvironment#1\newenvironment%
    \@lastclaim\else#1\fi}\noindent\ignorespaces}
   {This proves~\expandafter\ref{\@thisclaim}.\par\vskip 1ex\relax}
   {\let\@oldclaim\@thisclaim}
   {\let\@thisclaim\@oldclaim}
 \newcommand\blfootnote{\xdef\@thefnmark{}\@footnotetext}
 \makeatother
\renewcommand{\qed}{$\Box$}
\renewenvironment{proof}%
  {
  \noindent{\bf Proof.}\ }%
  {\hfill\qed\par\smallskip}
  {\setcounter{claim}{0}\noindent{\bf Proof of #1.}\ }%
  {\hfill\qed\par\bigskip}
\newcommand{\tp}{\mbox{\hspace*{0.1em}--\hspace*{0.1em}}}
\newcounter{fnote}
\newcommand\fnref[1]{$^{#1}$}
\newcommand\fntext[2][]{
\ifx\newenvironment#1\newenvironment%
\refstepcounter{fnote}$^{\arabic{fnote}}$\else$^{#1}$\fi~#2\par}
\newcommand\uline[1]{{\bf\textit{\color{NavyBlue}#1}}}
\renewcommand\bar[1]{\hspace*{0.05em}\overline{#1}\hspace*{0.05em}}
\newcommand{\centereq}[2][]{\par\noindent\relax\hfill{#2}\hfill{#1}\mbox{}\par}
\def\notdelta{\mbox{$\not\hspace*{-0.3em}\Delta$}}
\begin{document}

\title{Forbidden structure characterization of circular-arc graphs\\ and a
certifying recognition algorithm}

\author{
Mathew Francis\fnref{1}\thanks{Partially supported by the DST INSPIRE Faculty Award}~,
Pavol Hell\fnref{2}\thanks{Partially supported by a Discovery Grant from NSERC, and by the grant ERCCZ LL 1201}\,,
Juraj Stacho\fnref{3}\thanks{Partially supported by NSF grant DMS-1265803.}}
\date{\small
\fntext{Indian Statistical Institute (Chennai Centre), CIT Campus, Taramani, Chennai 600 113, India}
\fntext{School of Computing Science, Simon Fraser University, 8888 University Drive, Burnaby, Canada V5A 1S6}
\fntext{IEOR Department, Columbia University, 500 West 120th Street, New York, NY 10027, United States}\vspace{-5ex}}
\maketitle

%\blfootnote{{\bf Email addresses:} {\tt mathew@isichennai.res.in} (MF), {\tt pavol@sfu.ca} (PH), {\tt stacho@cs.toronto.edu} (JS) }

\begin{abstract}
A circular-arc graph is the intersection graph of arcs of a circle. It is a
well-studied graph model with numerous natural applications.  A certifying
algorithm is an algorithm that outputs a certificate, along with its answer (be
it positive or negative), where the certificate can be used to easily justify
the given answer.  While the recognition of circular-arc graphs has been known
to be polynomial since the 1980s, no polynomial-time certifying recognition
algorithm is known to date, despite such algorithms being found for many
subclasses of circular-arc graphs.  This is largely due to the fact that a
forbidden structure characterization of circular-arc graphs is not known, even
though the problem has been intensely studied since the seminal work of Klee in
the 1960s.

In this contribution, we settle this problem. We present the first forbidden
structure characterization of circular-arc graphs. Our obstruction has the form
of mutually avoiding walks in the graph. It naturally extends a similar
obstruction that characterizes interval graphs.  As a consequence, we give the
first polynomial-time certifying algorithm for the recognition of circular-arc
graphs.\medskip

\end{abstract} 

\section{Introduction}

A graph is a \uline{circular-arc graph} if it can be represented as the
intersection graph of arcs of a circle. Finding a forbidden subgraph
characterization of circular-arc graphs is a challenging open problem studied
since the late 1960s \cite{klee69,trotmoor76,tuck69,tuck70,tuck74}. Many partial
results toward this goal have been proposed over the years, but a full answer
remains elusive, capturing the interest of many researchers
\cite{bang94,bon09,fhh99,klee69,lin06,trotmoor76,tuck69,tuck70,tuck74}.  A
\uline{certifying algorithm} is an algorithm that outputs a certificate, along
with its answer (be it positive or negative), where the certificate can be used
to easily justify the given answer. For instance, a certifying algorithm for
testing planarity of a graph $G$ either outputs a plane embedding of $G$
(certifying that the graph \uline{is} planar) or outputs a subgraph of $G$
isomorphic to a subdivision of $K_{3,3}$ or $K_5$ (certifying that $G$ \uline{is
not} planar).  Similar certifying algorithms exist for recognition of interval
graphs \cite{int-perm-cert}, proper interval graphs \cite{propint-cert},
permutation graphs \cite{int-perm-cert}, chordal graphs \cite{book},
comparability graphs \cite{gallai}, and others.

The situation is markedly different for the recognition of circular-arc graphs.
Initially thought to be possibly NP-hard \cite{booth}, the problem has been
known to admit polynomial algorithms since the 1980s \cite{tuck80}. The
complexity of these algorithms has subsequently been improved
\cite{spinrad-eschen,hsu} and most recently a number of linear-time algorithms
\cite{kaplan-nussbaum,mcc03} have been proposed for the problem. Unfortunately,
all these algorithms suffer from the same issue.  They only produce a correct
output assuming one exists. Namely for an input graph $G$, they sometimes
construct an (incorrect) circular-arc representation that does not represent
$G$.  The algorithms check the correctness of the representation only at the end
of their execution; if it is not correct, they guarantee that there is no
representation but do not give a reason.  This has been a generally accepted
phenomenon of the recognition of circular-arc graphs.  One reason for this
situation is our limited understanding of smallest graphs that fail to be
circular-arc graphs, the so-called \uline{minimal forbidden obstructions}. In
fact, no minimal forbidden obstruction characterization has been proposed
despite many partial successes \cite{bon09, block-jour, tuck70, tuck74}.  In
order for the algorithm to provide a \uline{certificate} of the negative answer,
it must be understood what consitutes such a certificate (other than the
execution trace of the recognition algorithm).  This problem has been of
interest, since the very early work of Klee in the 1960s and has been intensely
studied over many years since.

In this paper, we propose the first answer to this problem. We describe a
structural obstruction that can be found in a graph if and only if the graph is
not a circular-arc graph. Our obstruction takes the form of a pair of mutually
avoiding walks that also avoid a fixed third vertex of the graph. 

To motivate this obstruction, we review analogous results.  A graph is an
\uline{interval graph} if it is the intersection graph of intervals of the real
line. Interval graphs are a subclass of circular-arc graphs. In the 1960s
Lekkerkerker and Boland introduced an \uline{asteroidal triple} to be a triple
of vertices of the graph such that between any two vertices of the triple there
exists a path where the third vertex has no neighbour on this path.  In
\cite{lekbol} they proved that a graph is an interval graph if and only if it
contains no asteroidal triple and no induced 4-cycle or 5-cycle.  A similar
concept was introduced 40 years later by Feder et al. \cite{adjusted}, who
defined an \uline{invertible pair} and proved that the absence of an invertible
pair also characterizes interval graphs. (This characterization generalizes to
directed graphs.)  The definition of invertible pair is more technical, but also
corresponds to walks with certain adjacency pattern.  An invertible pair is a
pair of vertices $\{u, v\}$ with equal-length walks from $u$ to $v$ and from $v$
to $u$, such that at every step, the current vertex of the first walk is not
adjacent to the next vertex of the second walk, and the same holds for some
other equal-length walks, with the role of $u$ and $v$ interchanged. (Note that
the graphs we consider have loops, and the walks can use them.)

In this work, we propose an analogous obstruction for circular arc graphs.  We
say that two vertices \uline{overlap} if they are adjacent but their
neighbourhoods are incomparable.  A vertex and a walk \uline{avoid} each other
if the vertex either has no neighbours on the walk, or it has neighbours on the
walk but it overlaps them, and if it overlaps both endpoints of an edge of the
walk, then the endpoints do not overlap each other.  Intuitively, if a vertex
$x$ avoids a walk, it prevents the walk from ``jumping over'' $x$ in the
representation.  This is easy to see if $x$ has no neighbours on the walk, while
if $x$ overlaps an endpoint of an edge $yz$ of the walk, or the endpoints $y,z$
do not overlap each other, then $y,z$ cannot be placed in the representation so
that they are on opposite sides of $x$.

An \uline{anchored invertible pair} is a pair of vertices $\{u, v\}$ together
with a third vertex (the ``anchor'') and with equal-length walks from $u$ to $v$
and from $v$ to $u$, such that at every step the current vertex of the first
(second) walk avoids the next (previous) edge of the second (first) walk, and
such that both walks avoid the anchor.

In an earlier paper~\cite{block-jour}, we proposed a simpler obstruction for
circular-arc graphs, called a {\em blocking quadruple}, which is more similar to
the asteroidal triple. The absence of blocking quadruples characterizes
circular-arc graphs amongst chordal graphs of independence number at most
five~\cite{block-jour}, and for several other special graph classes, for example
many of those studied in~\cite{bon09}. It can be checked that the presence
of a blocking quadruple in a graph implies also the presence of an anchored
invertible pair in the circular completion of the graph.

We illustrate these obstructions in Figure \ref{fig:example} for the biclaw
graph shown in Figure \ref{fig:biclaw}.  (Once again, note that the graphs we
consider have loops, and the walks can use them.)

\begin{figure}
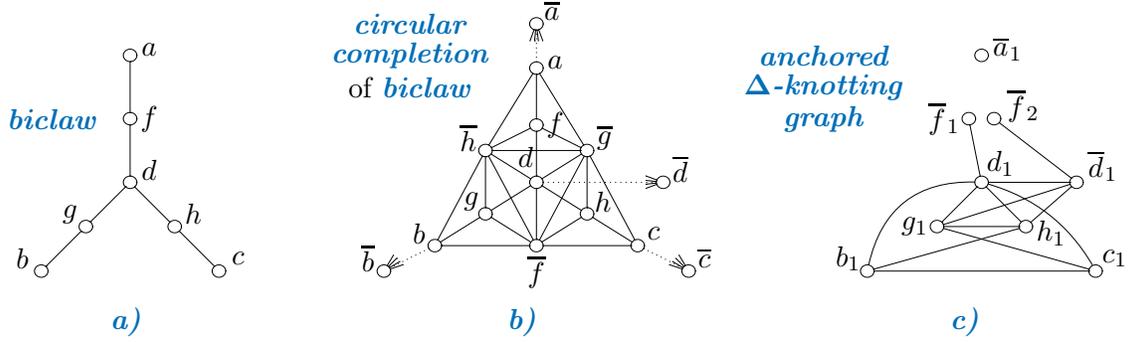

\centering
\begin{tabular}{c@{\qquad\quad}cc}
$\xy/r2pc/:
(0,0)*[o][F]{\phantom{s}}="x1";(0,1)*[o][F]{\phantom{s}}="x2";
(0,2)*[o][F]{\phantom{s}}="x3";(-0.7,-0.7)*[o][F]{\phantom{s}}="x4";
(0.7,-0.7)*[o][F]{\phantom{s}}="x5";(-1.4,-1.4)*[o][F]{\phantom{s}}="x6";
(1.4,-1.4)*[o][F]{\phantom{s}}="x7";{\ar@{-} "x1";"x2"};
{\ar@{-} "x2";"x3"};{\ar@{-} "x1";"x4"};{\ar@{-} "x1";"x5"};
{\ar@{-} "x4";"x6"};{\ar@{-} "x5";"x7"};
"x1"+(0.3,0.2)*{d};"x2"+(0.3,0)*{f};"x3"+(0.3,0.1)*{a};"x4"+(-0.25,0.2)*{g};
"x5"+(0.3,0.2)*{h};"x6"+(-0.3,0.2)*{b};"x7"+(0.3,0.2)*{c};
(-0.7,1)*[l]{\mbox{\uline{biclaw}}};
\endxy$
&
$\xy/r2pc/:
(0,0)*[o][F]{\phantom{s}}="x1";(0,0.9)*[o][F]{\phantom{s}}="x2";
(0,1.8)*[o][F]{\phantom{s}}="x3";(-0.8,-0.5)*[o][F]{\phantom{s}}="x4";
(0.8,-0.5)*[o][F]{\phantom{s}}="x5";(-1.6,-1)*[o][F]{\phantom{s}}="x6";
(1.6,-1)*[o][F]{\phantom{s}}="x7";(2,0)*[o][F]{\phantom{s}}="x1'";
(0,-1)*[o][F]{\phantom{s}}="x2'";(0,2.5)*[o][F]{\phantom{s}}="x3'";
(0.8,0.5)*[o][F]{\phantom{s}}="x4'";(-0.8,0.5)*[o][F]{\phantom{s}}="x5'";
(-2.4,-1.4)*[o][F]{\phantom{s}}="x6'";(2.4,-1.4)*[o][F]{\phantom{s}}="x7'";
{\ar@{-} "x1";"x2"};{\ar@{-} "x2";"x3"};{\ar@{-} "x1";"x4"};
{\ar@{-} "x1";"x5"};{\ar@{-} "x4";"x6"};{\ar@{-} "x5";"x7"};
"x1"+(-0.17,0.4)*{d};"x2"+(0.3,0)*{f};"x3"+(0.3,0.1)*{a};"x4"+(-0.22,0.2)*{g};
"x5"+(0.25,0.2)*{h};"x6"+(-0.25,0.2)*{b};"x7"+(0.25,0.2)*{c};
{\ar@{-} "x1'";"x1'"+(-0.27,-0.03)};{\ar@{-} "x1'";"x1'"+(-0.27,-0.1)};
{\ar@{-} "x1'";"x1'"+(-0.27,0.03)};{\ar@{-} "x1'";"x1'"+(-0.27,0.1)};
{\ar@{.} "x1'";"x1"};{\ar@{-} "x2'";"x1"};{\ar@{-} "x2'";"x4"};
{\ar@{-} "x2'";"x5"};{\ar@{-} "x2'";"x6"};{\ar@{-} "x2'";"x7"};
{\ar@{-} "x3'";"x3'"+(-0.03,-0.27)};{\ar@{-} "x3'";"x3'"+(-0.1,-0.27)};
{\ar@{-} "x3'";"x3'"+(0.03,-0.27)};{\ar@{-} "x3'";"x3'"+(0.1,-0.27)};
{\ar@{.} "x3'";"x3"};{\ar@{-} "x4'";"x1"};{\ar@{-} "x4'";"x2"};
{\ar@{-} "x4'";"x3"};{\ar@{-} "x4'";"x5"};{\ar@{-} "x4'";"x7"};
{\ar@{-} "x5'";"x1"};{\ar@{-} "x5'";"x2"};{\ar@{-} "x5'";"x3"};
{\ar@{-} "x5'";"x4"};{\ar@{-} "x5'";"x6"};{\ar@{-} "x6'";"x6'"+(0.2,0.2)};
{\ar@{-} "x6'";"x6'"+(0.25,0.15)};{\ar@{-} "x6'";"x6'"+(0.15,0.25)};
{\ar@{-} "x6'";"x6'"+(0.3,0.1)};{\ar@{.} "x6'";"x6"};
{\ar@{-} "x7'";"x7'"+(-0.2,0.2)};{\ar@{-} "x7'";"x7'"+(-0.25,0.15)};
{\ar@{-} "x7'";"x7'"+(-0.15,0.25)};{\ar@{-} "x7'";"x7'"+(-0.3,0.1)};
{\ar@{.} "x7'";"x7"};{\ar@{-} "x2'";"x4'"};{\ar@{-} "x2'";"x5'"};
{\ar@{-} "x4'";"x5'"};"x1'"+(0.27,0.2)*{\bar d};"x2'"+(0.0,-0.42)*{\bar f};
"x3'"+(0.25,0.2)*{\bar a};"x4'"+(0.27,0.2)*{\bar g};"x5'"+(-0.27,0.2)*{\bar h};
"x6'"+(-0.25,0.2)*{\bar b};"x7'"+(0.25,0.2)*{\bar c};
(-2,2.5)*{\mbox{\uline{circular}}};(-2,2)*{\mbox{\uline{completion}}};
(-2,1.5)*{\mbox{of \uline{biclaw}}};
\endxy$
&
$\xy/r2pc/:
(0,0)*[o][F]{\phantom{s}}="x1";(1.5,0)*[o][F]{\phantom{s}}="x1'";
(-0.2,1)*[o][F]{\phantom{s}}="x2'1";(0.2,1)*[o][F]{\phantom{s}}="x2'2";
(0,2)*[o][F]{\phantom{s}}="x3";(-0.7,-0.7)*[o][F]{\phantom{s}}="x4";
(0.7,-0.7)*[o][F]{\phantom{s}}="x5";(-1.8,-1.4)*[o][F]{\phantom{s}}="x6";
(1.8,-1.4)*[o][F]{\phantom{s}}="x7";{\ar@{-} "x1";"x2'1"};
{\ar@{-} "x1'";"x2'2"};{\ar@{-} "x1";"x4"};{\ar@{-} "x1";"x5"};
{\ar@{-} "x6";"x7"};{\ar@{-} "x1";"x1'"};{\ar@{-}@/^0.5pc/ "x1";"x7"};
{\ar@{-}@/_1pc/ "x1";"x6"};{\ar@{-} "x5";"x6"};{\ar@{-} "x4";"x7"};
{\ar@{-} "x1'";"x4"};{\ar@{-} "x1'";"x5"};{\ar@{-} "x4";"x5"};
"x1"+(0.3,0.3)*{d_1};"x1'"+(0.4,0.25)*{\bar d_1};"x2'1"+(-0.4,0)*{\bar f_1};
"x2'2"+(0.45,0.2)*{\bar f_2};"x3"+(0.4,0.1)*{\bar a_1};"x4"+(-0.35,0)*{g_1};
"x5"+(0.4,-0.1)*{h_1};"x6"+(-0.3,0.2)*{b_1};"x7"+(0.3,0.2)*{c_1};
(-2.5,2)*{\mbox{\uline{anchored}}};(-2.5,1.5)*{\mbox{\uline{$\bm\Delta$-knotting}}};
(-2.5,1)*{\mbox{\uline{graph}}};
\endxy$\medskip\\
\uline{a)}& \uline{b)}& \qquad\uline{c)}
\end{tabular}
\caption{\uline{a)} Example graph (biclaw), \uline{b)} its circular completion
-- vertices $\bar a$, $\bar b$, $\bar c$, $\bar d$ are pairwise adjacent and
adjacent to all other vertices except for $a$, $b$, $c$, $d$ respectively
(dotted edges=\mbox{non-edges}) \uline{c)} an anchored $\Delta$-knotting graph
of the completion -- with vertex $f$ as anchor (graph not
bipartite)\label{fig:biclaw}}
\end{figure}

\begin{figure}[ht!]
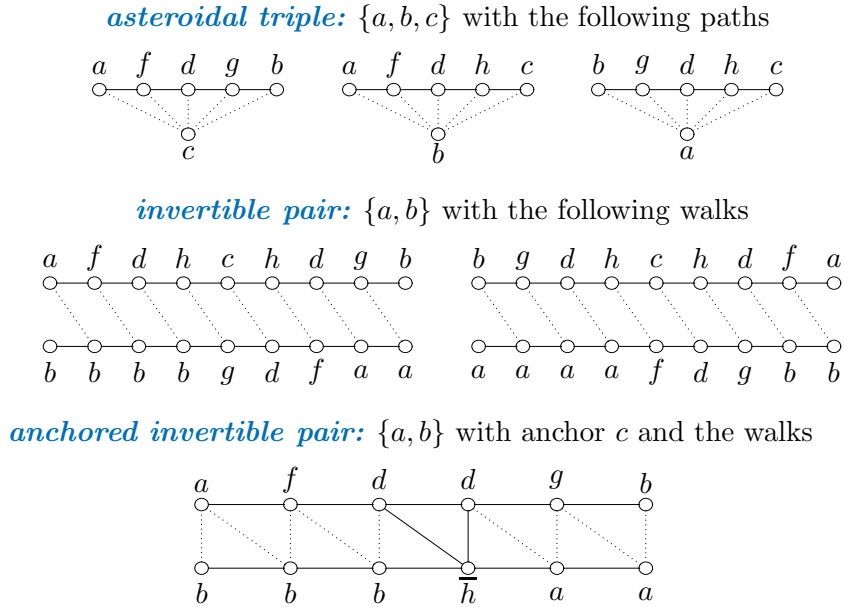
\centering
\parbox{0.7\textwidth}{\centering
\uline{asteroidal triple:} $\{a,b,c\}$ with the following paths\medskip

$\xy/r2pc/:
(0,0)*[o][F]{\phantom{s}}="x1";(0.7,0)*[o][F]{\phantom{s}}="x2";
(1.4,0)*[o][F]{\phantom{s}}="x3";(2.1,0)*[o][F]{\phantom{s}}="x4";
(2.8,0)*[o][F]{\phantom{s}}="x5";(1.4,-0.7)*[o][F]{\phantom{s}}="x6";
{\ar@{-} "x1";"x2"};{\ar@{-} "x2";"x3"};{\ar@{-} "x3";"x4"};
{\ar@{-} "x4";"x5"};{\ar@{.} "x1";"x6"};{\ar@{.} "x2";"x6"};
{\ar@{.} "x3";"x6"};{\ar@{.} "x4";"x6"};{\ar@{.} "x5";"x6"};
"x1"+(0,0.35)*{a};"x2"+(0,0.4)*{f};"x3"+(0,0.4)*{d};"x4"+(0,0.35)*{g};
"x5"+(0,0.4)*{b};"x6"+(0,-0.3)*{c};
\endxy$\qquad
$\xy/r2pc/:
(0,0)*[o][F]{\phantom{s}}="x1";(0.7,0)*[o][F]{\phantom{s}}="x2";
(1.4,0)*[o][F]{\phantom{s}}="x3";(2.1,0)*[o][F]{\phantom{s}}="x4";
(2.8,0)*[o][F]{\phantom{s}}="x5";(1.4,-0.7)*[o][F]{\phantom{s}}="x6";
{\ar@{-} "x1";"x2"};{\ar@{-} "x2";"x3"};{\ar@{-} "x3";"x4"};
{\ar@{-} "x4";"x5"};{\ar@{.} "x1";"x6"};{\ar@{.} "x2";"x6"};
{\ar@{.} "x3";"x6"};{\ar@{.} "x4";"x6"};{\ar@{.} "x5";"x6"};
"x1"+(0,0.35)*{a};"x2"+(0,0.4)*{f};"x3"+(0,0.4)*{d};
"x4"+(0,0.4)*{h};"x5"+(0,0.35)*{c};"x6"+(0,-0.3)*{b};
\endxy$\qquad
$\xy/r2pc/:
(0,0)*[o][F]{\phantom{s}}="x1";(0.7,0)*[o][F]{\phantom{s}}="x2";
(1.4,0)*[o][F]{\phantom{s}}="x3";(2.1,0)*[o][F]{\phantom{s}}="x4";
(2.8,0)*[o][F]{\phantom{s}}="x5";(1.4,-0.7)*[o][F]{\phantom{s}}="x6";
{\ar@{-} "x1";"x2"};{\ar@{-} "x2";"x3"};{\ar@{-} "x3";"x4"};
{\ar@{-} "x4";"x5"};{\ar@{.} "x1";"x6"};{\ar@{.} "x2";"x6"};
{\ar@{.} "x3";"x6"};{\ar@{.} "x4";"x6"};{\ar@{.} "x5";"x6"};
"x1"+(0,0.4)*{b};"x2"+(0,0.4)*{g};"x3"+(0,0.4)*{d};
"x4"+(0,0.4)*{h};"x5"+(0,0.35)*{c};"x6"+(0,-0.3)*{a};
\endxy$
}
\bigskip

\parbox{0.8\textwidth}{\centering
\uline{invertible pair:} $\{a,b\}$ with the following walks\medskip

$\xy/r2pc/:
(0,0)*[o][F]{\phantom{s}}="x1";(0.7,0)*[o][F]{\phantom{s}}="x2";
(1.4,0)*[o][F]{\phantom{s}}="x3";(2.1,0)*[o][F]{\phantom{s}}="x4";
(2.8,0)*[o][F]{\phantom{s}}="x5";(3.5,0)*[o][F]{\phantom{s}}="x6";
(4.2,0)*[o][F]{\phantom{s}}="x7";(4.9,0)*[o][F]{\phantom{s}}="x8";
(5.6,0)*[o][F]{\phantom{s}}="x9";(0,-1)*[o][F]{\phantom{s}}="y1";
(0.7,-1)*[o][F]{\phantom{s}}="y2";(1.4,-1)*[o][F]{\phantom{s}}="y3";
(2.1,-1)*[o][F]{\phantom{s}}="y4";(2.8,-1)*[o][F]{\phantom{s}}="y5";
(3.5,-1)*[o][F]{\phantom{s}}="y6";(4.2,-1)*[o][F]{\phantom{s}}="y7";
(4.9,-1)*[o][F]{\phantom{s}}="y8";(5.6,-1)*[o][F]{\phantom{s}}="y9";
"x1"+(0,0.35)*{a};"x2"+(0,0.4)*{f};"x3"+(0,0.4)*{d};
"x4"+(0,0.4)*{h};"x5"+(0,0.35)*{c};"x6"+(0,0.4)*{h};
"x7"+(0,0.4)*{d};"x8"+(0,0.35)*{g};"x9"+(0,0.4)*{b};
"y1"+(0,-0.4)*{b};"y2"+(0,-0.4)*{b};"y3"+(0,-0.4)*{b};
"y4"+(0,-0.4)*{b};"y5"+(0,-0.45)*{g};"y6"+(0,-0.4)*{d};
"y7"+(0,-0.4)*{f};"y8"+(0,-0.4)*{a};"y9"+(0,-0.4)*{a};
{\ar@{-} "x1";"x2"};{\ar@{-} "x2";"x3"};{\ar@{-} "x3";"x4"};
{\ar@{-} "x4";"x5"};{\ar@{-} "x5";"x6"};{\ar@{-} "x6";"x7"};
{\ar@{-} "x7";"x8"};{\ar@{-} "x8";"x9"};{\ar@{-} "y1";"y2"};
{\ar@{-} "y2";"y3"};{\ar@{-} "y3";"y4"};{\ar@{-} "y4";"y5"};
{\ar@{-} "y5";"y6"};{\ar@{-} "y6";"y7"};{\ar@{-} "y7";"y8"};
{\ar@{-} "y8";"y9"};{\ar@{.} "x1";"y2"};{\ar@{.} "x2";"y3"};
{\ar@{.} "x3";"y4"};{\ar@{.} "x4";"y5"};{\ar@{.} "x5";"y6"};
{\ar@{.} "x6";"y7"};{\ar@{.} "x7";"y8"};{\ar@{.} "x8";"y9"};
\endxy$\qquad
$\xy/r2pc/:
(0,0)*[o][F]{\phantom{s}}="x1";(0.7,0)*[o][F]{\phantom{s}}="x2";
(1.4,0)*[o][F]{\phantom{s}}="x3";(2.1,0)*[o][F]{\phantom{s}}="x4";
(2.8,0)*[o][F]{\phantom{s}}="x5";(3.5,0)*[o][F]{\phantom{s}}="x6";
(4.2,0)*[o][F]{\phantom{s}}="x7";(4.9,0)*[o][F]{\phantom{s}}="x8";
(5.6,0)*[o][F]{\phantom{s}}="x9";(0,-1)*[o][F]{\phantom{s}}="y1";
(0.7,-1)*[o][F]{\phantom{s}}="y2";(1.4,-1)*[o][F]{\phantom{s}}="y3";
(2.1,-1)*[o][F]{\phantom{s}}="y4";(2.8,-1)*[o][F]{\phantom{s}}="y5";
(3.5,-1)*[o][F]{\phantom{s}}="y6";(4.2,-1)*[o][F]{\phantom{s}}="y7";
(4.9,-1)*[o][F]{\phantom{s}}="y8";(5.6,-1)*[o][F]{\phantom{s}}="y9";
"x1"+(0,0.4)*{b};"x2"+(0,0.35)*{g};"x3"+(0,0.4)*{d};"x4"+(0,0.4)*{h};
"x5"+(0,0.35)*{c};"x6"+(0,0.4)*{h};"x7"+(0,0.4)*{d};"x8"+(0,0.4)*{f};
"x9"+(0,0.35)*{a};"y1"+(0,-0.4)*{a};"y2"+(0,-0.4)*{a};"y3"+(0,-0.4)*{a};
"y4"+(0,-0.4)*{a};"y5"+(0,-0.4)*{f};"y6"+(0,-0.4)*{d};"y7"+(0,-0.45)*{g};
"y8"+(0,-0.4)*{b};"y9"+(0,-0.4)*{b};{\ar@{-} "x1";"x2"};{\ar@{-} "x2";"x3"};
{\ar@{-} "x3";"x4"};{\ar@{-} "x4";"x5"};{\ar@{-} "x5";"x6"};
{\ar@{-} "x6";"x7"};{\ar@{-} "x7";"x8"};{\ar@{-} "x8";"x9"};
{\ar@{-} "y1";"y2"};{\ar@{-} "y2";"y3"};{\ar@{-} "y3";"y4"};
{\ar@{-} "y4";"y5"};{\ar@{-} "y5";"y6"};{\ar@{-} "y6";"y7"};
{\ar@{-} "y7";"y8"};{\ar@{-} "y8";"y9"};{\ar@{.} "x1";"y2"};
{\ar@{.} "x2";"y3"};{\ar@{.} "x3";"y4"};{\ar@{.} "x4";"y5"};
{\ar@{.} "x5";"y6"};{\ar@{.} "x6";"y7"};{\ar@{.} "x7";"y8"};{\ar@{.} "x8";"y9"};
\endxy$
}\bigskip

\centering\parbox{0.7\textwidth}{
\uline{anchored invertible pair:} $\{a,b\}$ with anchor $c$ and the walks\medskip

\centering
$\xy/r4pc/:
(0,0)*[o][F]{\phantom{s}}="x1";(0.7,0)*[o][F]{\phantom{s}}="x2";
(1.4,0)*[o][F]{\phantom{s}}="x3";(2.1,0)*[o][F]{\phantom{s}}="x4";
(2.8,0)*[o][F]{\phantom{s}}="x5";(3.5,0)*[o][F]{\phantom{s}}="x6";
(0,-0.5)*[o][F]{\phantom{s}}="y1";
(0.7,-0.5)*[o][F]{\phantom{s}}="y2";(1.4,-0.5)*[o][F]{\phantom{s}}="y3";
(2.1,-0.5)*[o][F]{\phantom{s}}="y4";(2.8,-0.5)*[o][F]{\phantom{s}}="y5";
(3.5,-0.5)*[o][F]{\phantom{s}}="y6";
"x1"+(0,0.15)*{a};"x1"+(0.35,0.2)*{};"x2"+(0,0.2)*{f};
"x3"+(0,0.2)*{d};"x3"+(0.35,0.2)*{};"x4"+(0,0.2)*{d};
"x4"+(0.35,0.2)*{};"x5"+(0,0.2)*{g};"x6"+(0,0.17)*{b};
"x6"+(0.35,0.2)*{};"y1"+(0,-0.2)*{b};
"y1"+(0.35,-0.2)*{};"y2"+(0,-0.2)*{b};"y2"+(0.35,-0.2)*{};
"y3"+(0,-0.2)*{b};"y3"+(0.35,-0.2)*{};"y4"+(0,-0.17)*{\bar h};
"y4"+(0.35,-0.2)*{};"y5"+(0,-0.2)*{a};"y5"+(0.35,-0.2)*{};
"y6"+(0,-0.2)*{a};"y6"+(0.35,-0.2)*{};
{\ar@{-} "x1";"x2"};{\ar@{-} "x2";"x3"};{\ar@{-} "x3";"x4"};
{\ar@{-} "x4";"x5"};{\ar@{-} "x5";"x6"};
{\ar@{-} "y1";"y2"};{\ar@{-} "y2";"y3"};{\ar@{-} "y3";"y4"};
{\ar@{-} "y4";"y5"};{\ar@{-} "y5";"y6"};
{\ar@{.} "x1";"y2"};{\ar@{.} "x1";"y1"};{\ar@{.} "x2";"y3"};
{\ar@{.} "y2";"x2"};{\ar@{-} "x3";"y4"};{\ar@{.} "y3";"x3"};
{\ar@{.} "x4";"y5"};{\ar@{-} "x4";"y4"};{\ar@{.} "x5";"y6"};
{\ar@{.} "y5";"x5"};{\ar@{.} "y6";"x6"};
\endxy$
}

\caption{Examples of various obstructions in the biclaw.\label{fig:example}}
\end{figure}

\begin{figure}
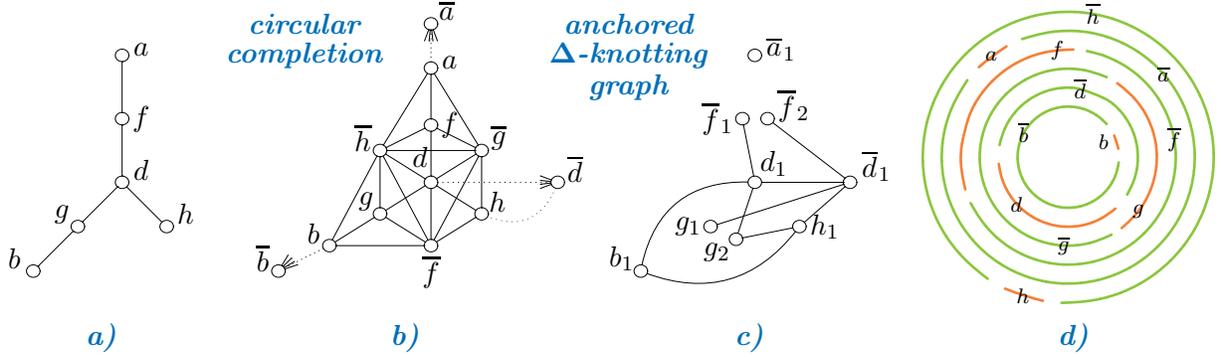
\centering
\begin{tabular}{@{}cc@{\hspace*{-1.2em}}cc@{}}
$\xy/r2pc/:
(0,0)*[o][F]{\phantom{s}}="x1";(0,1)*[o][F]{\phantom{s}}="x2";
(0,2)*[o][F]{\phantom{s}}="x3";(-0.7,-0.7)*[o][F]{\phantom{s}}="x4";
(0.7,-0.7)*[o][F]{\phantom{s}}="x5";(-1.4,-1.4)*[o][F]{\phantom{s}}="x6";
{\ar@{-} "x1";"x2"};{\ar@{-} "x2";"x3"};{\ar@{-} "x1";"x4"};
{\ar@{-} "x1";"x5"};{\ar@{-} "x4";"x6"};"x1"+(0.3,0.2)*{d};
"x2"+(0.3,0)*{f};"x3"+(0.3,0.1)*{a};"x4"+(-0.25,0.2)*{g};
"x5"+(0.3,0.2)*{h};"x6"+(-0.3,0.2)*{b};
\endxy$
&
$\xy/r2pc/:
(0,0)*[o][F]{\phantom{s}}="x1";(0,0.9)*[o][F]{\phantom{s}}="x2";
(0,1.8)*[o][F]{\phantom{s}}="x3";(-0.8,-0.5)*[o][F]{\phantom{s}}="x4";
(0.8,-0.5)*[o][F]{\phantom{s}}="x5";(-1.6,-1)*[o][F]{\phantom{s}}="x6";
(2,0)*[o][F]{\phantom{s}}="x1'";(0,-1)*[o][F]{\phantom{s}}="x2'";
(0,2.5)*[o][F]{\phantom{s}}="x3'";(0.8,0.5)*[o][F]{\phantom{s}}="x4'";
(-0.8,0.5)*[o][F]{\phantom{s}}="x5'";(-2.4,-1.4)*[o][F]{\phantom{s}}="x6'";
{\ar@{-} "x1";"x2"};{\ar@{-} "x2";"x3"};{\ar@{-} "x1";"x4"};
{\ar@{-} "x1";"x5"};{\ar@{-} "x4";"x6"};"x1"+(-0.17,0.4)*{d};
"x2"+(0.3,0)*{f};"x3"+(0.3,0.1)*{a};"x4"+(-0.22,0.2)*{g};
"x5"+(0.25,0.2)*{h};"x6"+(-0.25,0.2)*{b};{\ar@{-} "x1'";"x1'"+(-0.27,-0.03)};
{\ar@{-} "x1'";"x1'"+(-0.27,-0.1)};{\ar@{-} "x1'";"x1'"+(-0.27,0.03)};
{\ar@{-} "x1'";"x1'"+(-0.27,0.1)};{\ar@{.} "x1'";"x1"};
{\ar@{.}@/^0.7pc/ "x1'";"x5"};{\ar@{-} "x2'";"x1"};
{\ar@{-} "x2'";"x4"};{\ar@{-} "x2'";"x5"};{\ar@{-} "x2'";"x6"};
{\ar@{-} "x3'";"x3'"+(-0.03,-0.27)};{\ar@{-} "x3'";"x3'"+(-0.1,-0.27)};
{\ar@{-} "x3'";"x3'"+(0.03,-0.27)};{\ar@{-} "x3'";"x3'"+(0.1,-0.27)};
{\ar@{.} "x3'";"x3"};{\ar@{-} "x4'";"x1"};{\ar@{-} "x4'";"x2"};
{\ar@{-} "x4'";"x3"};{\ar@{-} "x4'";"x5"};{\ar@{-} "x6'";"x6'"+(0.2,0.2)};
{\ar@{-} "x6'";"x6'"+(0.25,0.15)};{\ar@{-} "x6'";"x6'"+(0.15,0.25)};
{\ar@{-} "x6'";"x6'"+(0.3,0.1)};{\ar@{.} "x6'";"x6"};
{\ar@{-} "x2'";"x4'"};{\ar@{-} "x2'";"x5'"};{\ar@{-} "x4'";"x5'"};
{\ar@{-} "x5'";"x1"};{\ar@{-} "x5'";"x2"};{\ar@{-} "x5'";"x3"};
{\ar@{-} "x5'";"x4"};{\ar@{-} "x5'";"x6"};"x1'"+(0.27,0.2)*{\bar d};
"x2'"+(0.0,-0.42)*{\bar f};"x3'"+(0.25,0.2)*{\bar a};"x4'"+(0.27,0.2)*{\bar g};
"x5'"+(-0.27,0.2)*{\bar h};"x6'"+(-0.25,0.2)*{\bar b};
(-2,2.5)*{\mbox{\uline{circular}}};(-2,2)*{\mbox{\uline{completion}}};
\endxy$
&
$\xy/r2pc/:
(0,0)*[o][F]{\phantom{s}}="x1";(1.5,0)*[o][F]{\phantom{s}}="x1'";
(-0.2,1)*[o][F]{\phantom{s}}="x2'1";(0.2,1)*[o][F]{\phantom{s}}="x2'2";
(0,2)*[o][F]{\phantom{s}}="x3";(-0.7,-0.7)*[o][F]{\phantom{s}}="x4";
(-0.3,-0.9)*[o][F]{\phantom{s}}="x4'";(0.7,-0.7)*[o][F]{\phantom{s}}="x5";
(-1.8,-1.4)*[o][F]{\phantom{s}}="x6";{\ar@{-} "x1";"x2'1"};
{\ar@{-} "x1'";"x2'2"};{\ar@{-} "x1'";"x4"};{\ar@{-} "x1'";"x5"};
{\ar@{-} "x1";"x1'"};{\ar@{-}@/_1pc/ "x1";"x6"};{\ar@{-}@/^1pc/ "x5";"x6"};
{\ar@{-} "x1";"x4'"};{\ar@{-} "x4'";"x5"};"x1"+(0.3,0.3)*{d_1};
"x1'"+(0.4,0.25)*{\bar d_1};"x2'1"+(-0.4,0)*{\bar f_1};"x2'2"+(0.4,0.2)*{\bar f_2};
"x3"+(0.4,0.1)*{\bar a_1};"x4"+(-0.35,0)*{g_1};"x4'"+(-0.3,-0.2)*{g_2};
"x5"+(0.4,0)*{h_1};"x6"+(-0.3,0.2)*{b_1};(-2,2.5)*{\mbox{\uline{anchored}}};
(-2,2)*{\mbox{\uline{$\bm\Delta$-knotting}}};(-2,1.5)*{\mbox{\uline{graph}}};
\endxy$&
\raisebox{2ex}{$\xy/r2pc/:
{{\xypolygon46"A"{~:{(1.7,0):}~>{}~={80}{}}},{\xypolygon46"B"{~:{(2.0,0):}~>{}~={80}{}}},
{\xypolygon46"C"{~:{(2.3,0):}~>{}~={80}{}}},{\xypolygon46"D"{~:{(1.4,0):}~>{}~={80}{}}},
{\xypolygon48"E"{~:{(1.1,0):}~>{}~={80}{}}},{\xypolygon46"F"{~:{(0.8,0):}~>{}~={80}{}}}};
"A2";"A16"**[Orange][|1pt]\crv{"A3"& "A4"& "A5"& "A6"& "A7"& "A8"& "A9"& "A10"&
"A11"& "A13"& "A14"& "A15"}; "A17";"A1"**[LimeGreen][|1pt]\crv{"A18"& "A19"&
"A20"& "A21"& "A22"& "A23"& "A24"& "A25"& "A26"& "A27"& "A28"& "A29"& "A30"&
"A31"& "A32"& "A33"& "A34"& "A35"& "A36"& "A37"& "A38"& "A39"& "A40"& "A41"&
"A42"& "A43"& "A44"& "A45"& "A46"& }; "B6";"B8"**[Orange][|1pt]\crv{"B7"};
"B9";"B5"**[LimeGreen][|1pt]\crv{"B10"& "B11"& "B12"& "B13"& "B14"& "B15"&
"B16"& "B17"& "B18"& "B19"& "B20"& "B21"& "B22"& "B23"& "B24"& "B25"& "B26"&
"B27"& "B28"& "B29"& "B30"& "B31"& "B32"& "B33"& "B34"& "B35"& "B36"& "B37"&
"B38"& "B39"& "B40"& "B41"& "B42"& "B43"& "B44"& "B45"& "B46"& "B1"& "B2"& "B3"&
"B4"}; "E15";"E32"**[Orange][|1pt]\crv{"E15"& "E16"& "E17"& "E18"& "E19"& "E20"&
"E21"& "E22"& "E23"& "E24"& "E25"& "E26"& "E27"& "E28"& "E29"& "E30"& "E31"};
"E34";"E13"**[LimeGreen][|1pt]\crv{"E34"& "E35"& "E36"& "E37"& "E38"& "E39"&
"E40"& "E41"& "E42"& "E43"& "E44"& "E45"& "E46"& "E1"& "E2"& "E3"& "E4"& "E5"&
"E6"& "E7"& "E8"& "E9"& "E10"& "E11"& "E12"};
"C22";"C24"**[Orange][|1pt]\crv{"C23"};
"C25";"C21"**[LimeGreen][|1pt]\crv{"C26"& "C27"& "C28"& "C29"& "C30"& "C31"&
"C32"& "C33"& "C34"& "C35"& "C36"& "C37"& "C38"& "C39"& "C40"& "C41"& "C42"&
"C43"& "C44"& "C45"& "C46"& "C1"& "C2"& "C3"& "C4"& "C5"& "C6"& "C7"& "C8"&
"C9"& "C10"& "C11"& "C12"& "C13"& "C14"& "C15"& "C16"& "C17"& "C18"& "C19"&
"C20"}; "D30";"D44"**[Orange][|1pt]\crv{"D31"& "D32"& "D33"& "D34"& "D35"&
"D36"& "D37"& "D38"& "D39"& "D40"& "D41"& "D42"& "D43"};
"D45";"D29"**[LimeGreen][|1pt]\crv{"D46"& "D1"& "D2"& "D3"& "D4"& "D5"& "D6"&
"D7"& "D8"& "D9"& "D10"& "D11"& "D12"& "D13"& "D14"& "D15"& "D16"& "D17"& "D18"&
"D19"& "D20"& "D21"& "D22"& "D23"& "D24"& "D25"& "D26"& "D27"& "D28"};
"F38";"F40"**[Orange][|1pt]\crv{"F39"};
"F42";"F36"**[LimeGreen][|1pt]\crv{"F42"& "F43"& "F44"& "F45"& "F46"& "F1"&
"F2"& "F3"& "F4"& "F5"& "F6"& "F7"& "F8"& "F9"& "F10"& "F11"& "F12"& "F13"&
"F14"& "F15"& "F16"& "F17"& "F18"& "F19"& "F20"& "F21"& "F22"& "F23"& "F24"&
"F25"& "F26"& "F27"& "F28"& "F29"& "F30"& "F31"& "F32"& "F33"& "F34"& "F35"&
"F36"}; "A3"*{_f};"A38"*{_{\bar f}};"B7"*{_a};"B42"*{_{\bar
a}};"E20"*{_d};"E1"*{_{\bar d}};"C23"*{_h};"C1"*{_{\bar h}};
"D32"*{_g};"D25"*{_{\bar g}};"F39"+(-0.2,0)*{_b};"F10"*{_{\bar b}};
\endxy$}
\medskip\\\uline{a)}&\uline{b)}&\qquad\uline{c)}&\uline{d)}
\end{tabular}
\caption{\uline{a)} Another example graph, \uline{b)} its circular completion,
\uline{c)} an anchored $\Delta$-knotting graph of the completion -- anchor $f$
(graph is bipartite), \uline{d)} circular-arc representation of the
completion\label{fig:biclaw2}}
\end{figure}

The intuition for the anchored invertible pair being an obstruction comes from
the fact that there is a place on the circle (corresponding to the anchor) that
is avoided by the two mutually avoiding walks.  Since the two walks must avoid
the area of the anchor, they certify that a circular-arc representation is
impossible.  We note that the concept of avoidance used here is different (and
somewhat more technical) from what the authors of \cite{adjusted} used for
invertible pairs. This has a good reason -- in each case, the definition is
motivated by a standard ordering characterization. For interval graphs, the
ordering characterization is simple \cite{intorder}, while for circular arc
graphs we use a more complex ordering (based on \uline{normalized
representations} of \cite{hsu}) which directly motivates our definition of
avoidance.

One important idea that allowed us to find a forbidden obstruction
characterization is the introduction of ``circular completions''.  Specifically,
we add for any vertex $v$ another vertex $\bar v$, with the property that in any
potential representation the arcs $v$ and $\bar v$ are complementary, i.e., are
disjoint and cover the circle.  This allows us to directly use the idea of
\uline{arc flipping} on which the recognition algorithm of McConnell
\cite{mcc03} is based, but with the added advantage of capturing all possible
such flips at once. The corresponding machinery developed in \cite{mcc03} is
what drives our subsequent proofs.

As a consequence, we describe the first fully certifying polynomial-time
algorithm for recognition of circular-arc graphs.  Given the existing
(yes-certifying) recognition algorithms \cite{kaplan-nussbaum,mcc03} for the
problem, it suffices to certify a negative answer.  To this end, the algorithm
constructs an auxiliary graph, the so-called \uline{anchored
$\bm\Delta$-knotting graph} (in analogy to the knotting graph of Gallai
\cite{gallai}).  This graph must contain an odd cycle if the input is not a
circular-arc graph. Using the cycle, the algorithm produces an anchored
invertible pair and avoiding walks certifying this pair.

Similar algorithms based on the idea of knotting are known to recognize
cocomparability \cite{gallai} and asteroidal-triple-free graphs
\cite{path-order}. 

\section{Definitions}

We write $(u,v)$ to denote an ordered pair and $\{u,v\}$ for an unordered pair;
in case it is an edge, we abbreviate it to $uv$.  A graph $G=(V,E)$ has vertex
set $V=V(G)$ and edge set $E=E(G)$. As usual, the graphs considered here are
undirected and have no parallel edges.  However, all graphs here are assumed to
have a \uline{loop} $uu\in E$ at every vertex $u$.  This is consistent with
intersection representations (every non-empty set intersects itself). For $X
\subseteq V$,
we write $G[X]$ to denote the subgraph of $G$ induced by $X$, i.e. the graph
with vertex set $X$ where two vertices are adjacent if and only if they are
adjacent in $G$.  We write $N(u)$ to denote the (open) neighbourhood of $u$,
i.e. vertices $v\neq u$ such that $uv\in E$. We write $N[u]$ for the set
$N(u)\cup\{u\}$, the closed neighbourhood of $u$.  When dealing with different
graphs, we write $N_G(u)$ and $N_G[u]$ to indicate the neighbourhoods of $u$ in
the graph $G$.  A vertex $u$ is \uline{universal} if $N[u]=V$.  Vertices $u,v$
are \uline{true twins} if $N[u]=N[v]$.

We write $P=x_1\tp x_2\tp\cdots\tp x_k$ to denote a walk in $G$, i.e. the
sequence of (not necessarily distinct) vertices $x_1,x_2,\ldots,x_k\in V$ where
$x_ix_{i+1}\in E$ for all $1\leq i<k$. Some vertices may repeat on the walk,
even consecutive vertices can be identical, since every vertex has a loop.

\section{Overview of the result}

The main result of the paper is described below as Theorem \ref{thm:main}. It
characterizes circular-arc graphs as graphs that do not contain a specific
structural obstruction. The obstruction is defined~as~follows.

We say that vertex $x$ \uline{overlaps} vertex $y$ (or that $x,y$ overlap each
other) if $x,y$ are adjacent but have incomparable neighbourhoods.  We say that
a vertex $z$ and a walk $P$ \uline{avoid} each other if\smallskip

\begin{compactenum}[(i)]
\item every neighbour of $z$ on $P$ overlaps $z$, and
\item  if $z$ overlaps both endpoints $x$, $y$ of an edge $xy$ on $P$, then
$x$, $y$ do not overlap each other.
\end{compactenum}\smallskip

\noindent If $z$ avoids a walk $x\tp y$, we simply say that $z$ avoids the
edge $xy$.  (Note that $x=y$ is allowed here. Namely, the definition implies
that $z$ avoids a loop $xx$ iff $z$ is non-adjacent to $x$ or overlaps $x$.)
\smallskip

We say that a walk $P=x_1\tp x_2\tp\cdots\tp x_k$ \uline{avoids} a walk $Q=y_1
\tp y_2\tp\cdots\tp y_k$ if for each $1\leq i<k$, the vertex $x_i$ avoids the
edge $y_iy_{i+1}$, and the vertex $y_{i+1}$ avoids the edge $x_ix_{i+1}$. Note
that we may assume that for all $i$, either $x_i=x_{i+1}$ or $y_i=y_{i+1}$. If
this is not so, for each $i$ for which $x_i\neq x_{i+1}$ and $y_i\neq y_{i+1}$,
we replace the subwalks $x_i\tp x_{i+1}$ in $P$ and $y_{i}\tp y_{i+1}$ in $Q$
by walks $x_i\tp x_i\tp x_{i+1}$ and $y_i\tp y_{i+1}\tp y_{i+1}$ respectively
to obtain two new walks $P'$ and $Q'$ respectively. The walk $P'$ avoids the
walk $Q'$, since for each $i$, $x_i$ avoids $y_iy_{i+1}$ and $y_{i+1}$ avoids
$x_ix_{i+1}$. For this, note further that this implies that $x_i$ overlaps or is
non-adjacent to $y_{i+1}$ and thus $x_i$ avoids the loop $y_{i+1}y_{i+1}$.
Likewise $y_{i+1}$ avoids the loop $x_ix_i$. It is not hard to see that the walk
$Q'$ avoids the walk $P'$ as well.

Whenever we say that a walk $P$ avoids a walk $Q$, we shall assume that $P$ and
$Q$ satisfy the above assumption, and therefore
%It follows from this observation that if $P$ avoids $Q$, then $Q$ avoids $P$,
%and
we shall simply say that $P$ and $Q$ avoid each other.

For distinct vertices $x,y,z$, we say that $x,y$ is a \uline{$\bm z$-invertible
pair}, if there exists an $xy$-walk and a $yx$-walk such that the two walks
avoid each other and both avoid $z$.  We also say that $x,y$ is an invertible
pair anchored at $z$, or simply an \uline{anchored invertible pair}.

We also need to introduce the concept of \uline{circular completion}. Roughly
speaking, the circular completion of $G$ is constructed from $G$ by adding
additional vertices so that every vertex is paired-up with a ``complementary''
vertex. We defer the precise definition until a later section, since it uses
concepts introduced there.

Now we are ready to state the main results of this paper.

\begin{theorem}\label{thm:main}
Let $G$ be a graph. Assume $G$ has no universal vertices or true twins.

Then the following are equivalent.
\begin{compactenum}[\rm (I)]
\item $G$ is a circular-arc graph.
\setreflabel{\rm(\Roman{enumi})}\label{enum:I}
\item The circular completion of $G$ contains no anchored invertible pair.
\setreflabel{\rm(\Roman{enumi})}\label{enum:II}
\end{compactenum}
\end{theorem}

\begin{theorem}\label{thm:certif}
There exists an $O(n^4)$ time certifying algorithm for recognition of
circular-arc graphs.
\end{theorem}

The rest of the paper is organized as follows. In \S \ref{sec:prelim} we
introduce key concepts and state relevant results from the literature.  In \S
\ref{sec:completions}, we define circular completion and show how to obtain one
for a given graph.   In \S \ref{sec:charac}, we prove Theorem \ref{thm:main}.
In \S \ref{sec:certif} we discuss the resulting certifying algorithm for
recognition of circular-arc graphs using the anchored $\Delta$-knotting graph
described in \S \ref{sec:knot}. 

%We conclude the paper in \S \ref{sec:conclusion} by outlining further extensions
%of this work.\vspace{-1ex}

\section{Preliminaries}\label{sec:prelim}

Let $G=(V,E)$ be a graph. Throughout the paper (unless said otherwise) we shall
assume that $G$ has no universal vertices or true twins. The question of
representability by circular arcs is unaffected by this.  Also, recall that we
assume that there is a \uline{loop} $uu\in E$ at every vertex $u\in V$.

When traversing around a circle in the clockwise direction, we first encounter
the \uline{left endpoint} of a circular arc, then its \uline{interior}, and then
its \uline{right endpoint}.  For points $p,q,r$ on the circle, we write $p<q<r$
to indicate that the points $p,q,r$ appear in this order in a full traversal of
the circle in the clockwise direction starting from the point $p$.  Similarly,
we write $p_1<p_2<\ldots<p_k$ for the clockwise order of points
$p_1,p_2,\ldots,p_k$. We write $[p,q]$ for the clockwise arc from $p$ to $q$.

\begin{definition}
A \uline{circular-arc representation} of a graph $G=(V,E)$ assigns to each $u\in
V$ a circular arc $[l_u,r_u]$ where $l_u$ and $r_u$ are the left and right
endpoints of this arc such that all endpoints of arcs are distinct, and $uv\in
E$ iff $[l_u,r_u]$ intersects $[l_v,r_v]$.
\end{definition}

\begin{definition}[Spanning pair]
We say that $\{u,v\}$ is a \uline{spanning pair} of $G$ if

\begin{compactenum}[\qquad(C1)]
\item for all $x\in V\setminus N[v]$, we have $N[x]\subseteq N[u]$,
\setreflabel{(C\theenumi)}\label{enum:c1}
\item for all $y\in V\setminus N[u]$, we have $N[y]\subseteq N[v]$.
\setreflabel{(C\theenumi)}\label{enum:c2}
\end{compactenum}
\end{definition}

\begin{definition}[Edge types]\label{def:edgetypes}
We say that an edge $uv\in E$ is
\begin{compactitem}
\item an \uline{inclusion edge} if $N[u]\subseteq N[v]$ or $N[u]\supseteq N[v]$,
\item an \uline{overlap edge} otherwise.
\end{compactitem}
\smallskip

\noindent For overlap edges $uv$, we further distinguish two subtypes:
\begin{compactitem}
\item a \uline{2-overlap edge} if $\{u,v\}$ is a spanning pair,
\item a \uline{1-overlap edge} otherwise.
\end{compactitem}

\end{definition}

Note that every edge $uv$ is classified either as an inclusion edge, a 1-overlap
edge or a 2-overlap edge. Also, note that a loop $uu\in E$ is classified by the
above as an inclusion edge.

%\subsection{Normalized representations}\label{sec:normrepresentations}

\begin{definition}[Normalized representation]
Let $X\subseteq V$.  A circular-arc representation $\{[l_u,r_u]\}_{u\in X}$ of
$G[X]$ is said to be \uline{normalized with respect to $\bm G$} if for all
$uv\in E$ such that $u,v\in X$:\smallskip

\begin{compactenum}[\quad\rm(N1)]
\item the arc $[l_u,r_u]$ properly contains $[l_v,r_v]$ $\iff$
$N[u]\supseteq N[v]$ in $G$ \hfill($l_u<l_v<r_v<r_u$)
\setreflabel{\rm(N\theenumi)}\label{enum:n1}
\item $[l_u,r_u]\cup[l_v,r_v]$ covers the circle $\iff$ $uv$ is a
2-overlap edge of $G$\hfill($l_u<r_v<l_v<r_u$)
\setreflabel{\rm(N\theenumi)}\label{enum:n2}
\end{compactenum}\smallskip
If $X=V$, then the representation is said to be \uline{normalized} (for short).
\end{definition}

\begin{theorem}{\rm\cite{hsu}}\label{thm:hsu}
If $G$ is a circular-arc graph with no universal vertices or true twins, then
$G$ has a normalized circular-arc representation.
\end{theorem}

\begin{lemma}\label{lem:avoid}
Let $\{[l_u,r_u]\}_{u\in V}$ be a normalized circular-arc representation of
$G=(V,E)$. Let $xy\in E$ and $z,w\in V$ be such that $l_x<l_z<l_y<l_w$.
Then at least one of $z$, $w$ does not avoid the edge $xy$.
\end{lemma}

\begin{proof}
For contradiction, assume that both $z$ and $w$ do avoid $xy$.  Since $xy\in E$,
it is not the case that $l_x<r_x<l_y<r_y$.  Thus by symmetry, we may assume that
$l_x<l_y<r_x$.  Therefore $l_x<l_z<r_x$.  This implies $xz\in E$. Thus since $z$
avoids $xy$, the edge $xz$ is  an overlap edge. In particular,
$N[z]\not\subseteq N[x]$.  Therefore $l_z<r_x<r_z$. Altogether, we have
$l_z<l_y<r_x<r_z$. Thus $yz\in E$ and again $yz$ is an overlap edge. In
particular, $N[y]\not\subseteq N[z]$. This implies $l_y<r_z<r_y$.  From this we
deduce that $l_y<r_x<r_y$. This means, however, that neither $[l_x,r_x]$
contains $[l_y,r_y]$, nor $[l_y,r_y]$ contains $[l_x,r_x]$.  Thus $xy$ is also
an overlap edge. But all three edges $xy, xz, yz$ are now overlap edges, which
is impossible, since $z$ avoids the edge $xy$.  This completes the proof.
\end{proof}

\section{Circular pairs and the circular completion}\label{sec:completions}

\begin{definition}
We say that $\{u,v\}$ is a \uline{circular pair} if $\{u,v\}$ is a spanning pair and
$uv\not\in E$.\smallskip

A vertex $u$ is \uline{circularly-paired} if there exists $v\in V$ such that
$\{u,v\}$ is a circular pair.\smallskip

We say that $G$ is \uline{circularly-paired} if every $u\in V$ is
circularly-paired.\smallskip

A \uline{circular completion} of $G$ is a smallest circularly-paired graph $H$
such that

\begin{compactenum}[\quad(i)]
\item $H$ contains $G$ as an induced subgraph,
\item every edge $uv\in E(G)$ has the same type both in $G$ and in $H$.
\end{compactenum}
\end{definition}

Note that the notion of circular completion requires that $G$ has no universal
vertices (there is no ``complement'' of a universal vertex).  The following are
essential properties of circular pairs.

\begin{lemma}\label{lem:circ1}
Assume that $G$ has no universal vertices or true twins. Let $\{u,v\}$ be a circular pair.
\begin{compactenum}[\quad\rm (\thetheorem.1)]
\item $N[u]\supseteq N[w]$ if and only if $vw\not\in E$.
\setreflabel{\rm(\thetheorem.\theenumi)}\label{enum:circ1-0}
\item $N[u]\subseteq N[w]$ if and only if $vw$ is a 2-overlap edge.
\setreflabel{\rm(\thetheorem.\theenumi)}\label{enum:circ1-i}
\item $uw$ is a 1-overlap edge if and only if $vw$ is a 1-overlap edge.
\setreflabel{\rm(\thetheorem.\theenumi)}\label{enum:circ1-ii}
\item If $\{u,w\}$ is a circular pair, then $v=w$.
\setreflabel{\rm(\thetheorem.\theenumi)}\label{enum:circ1-iii}
\end{compactenum}
\end{lemma}

\begin{proof}
Since $\{u,v\}$ is a circular pair, $uv\not\in E$ and $\{u,v\}$ is a spanning
pair.

For \ref{enum:circ1-0}, if $vw\not\in E$, then $N[w]\subseteq N[u]$ by
\ref{enum:c1} applied to $\{u,v\}$. Conversely, if $N[w]\subseteq N[u]$, then
$vw\not\in E$, since $uv\not\in E$. This proves \ref{enum:circ1-0}.

For \ref{enum:circ1-i}, if $vw$ is a 2-overlap edge, then $\{v,w\}$ is a
spanning pair, and so $N[u]\subseteq N[w]$ by \ref{enum:c2} applied to $\{v,w\}$
since $uv\not\in E$. 

Conversely assume $N[u]\subseteq N[w]$.  If $vw\not\in E$, then $N[w]\subseteq
N[u]$ by \ref{enum:c1} applied to $\{u,v\}$. So $N[u]=N[w]$, which is impossible,
since $G$ has no true twins.  Therefore $vw\in E$.  Clearly, $N[w]\not\subseteq
N[v]$, since $u\in N[w]\setminus N[v]$. Suppose that $N[v]\subseteq N[w]$ and
recall that $N[u]\subseteq N[w]$.  Since $G$ has no universal vertex, there
exists a vertex $x\in V\setminus N[w]$. This means $x\not\in N[u]$ and $x\not\in
N[v]$ contrary to \ref{enum:c1}.  Therefore $N[v]\not\subseteq N[w]$ and so $vw$
is not an inclusion edge.  It remains to prove that $\{v,w\}$ is a spanning
pair.  For \ref{enum:c1}, consider a vertex $x\in V\setminus N[v]$. Then
$N[x]\subseteq N[u]$ by \ref{enum:c1} for $\{u,v\}$. Thus $N[x]\subseteq N[w]$,
since $N[u]\subseteq N[w]$.  For \ref{enum:c2}, consider a vertex $x\in
V\setminus N[w]$. Then $x\in V\setminus N[u]$, since $N[u]\subseteq N[w]$.  Thus
$N[x]\subseteq N[v]$ by \ref{enum:c2} for $\{u,v\}$.  This therefore verifies
conditions \ref{enum:c1}-\ref{enum:c2} for $\{v,w\}$ and so $vw$ is indeed a
2-overlap edge. This proves \ref{enum:circ1-i}.

For \ref{enum:circ1-ii}, suppose that $uw$ is a 1-overlap edge. By
\ref{enum:circ1-0}, we see that $vw$ must be an edge. Clearly,
$N[w]\not\subseteq N[v]$, since $u\in N[w]\setminus N[v]$.  If $N[v]\subseteq
N[w]$, then we would deduce by part \ref{enum:circ1-i} that $uw$ is a 2-overlap
edge, a contradiction.  Thus $vw$ is not an inclusion edge.  If $vw$ were a
2-overlap edge, then $N[u]\subseteq N[w]$ by \ref{enum:circ1-i} and $uw$ would
be an inclusion edge, a contradiction.  Therefore  $vw$ is not a 2-overlap edge,
and so it must be a 1-overlap edge. This proves \ref{enum:circ1-ii}.

Finally, for \ref{enum:circ1-iii}, if $\{u,w\}$ is a circular pair, then
$uw\not\in E$ and thus $N[w]\subseteq N[v]$ by \ref{enum:c2} applied to
$\{u,v\}$. Likewise, $N[v]\subseteq N[w]$ by \ref{enum:c2} applied to $\{u,w\}$,
since $uv\not\in E$.  But now $N[v]=N[w]$, impossible since $G$ has no true
twins.

This concludes the proof.
\end{proof}

\begin{lemma}\label{lem:normalized}
Assume that $G$ has no universal vertices or true twins. Suppose there exists
a set $X\subseteq V$ such that
\begin{compactenum}[(i)]
\item each vertex in $V\setminus X$ is circularly-paired with a vertex in $X$,
and
\item $G[X]$ admits a circular-arc representation $\{[l_u,r_u]\}_{u\in V}$ which
is normalized with respect to $G$.
\end{compactenum}
Then $G$ is a circular-arc graph.
\end{lemma}

\begin{proof}
By (ii), let $\{[l_u,r_u]\}_{u\in X}$ denote a circular-arc representation of
$G[X]$ normalized with respect to $G$. By (i), for each $v\in V\setminus X$,
there exists vertex $u\in X$ such that $\{u,v\}$ is a circular pair; by
\ref{enum:circ1-iii} such pair $\{u,v\}$ is unique (both with respect to $u$ and
to $v$), since $G$ has no true twins. Assign to $v$ the complement of the arc
$[l_u,r_u]$, and let $\bar u$ refer to the vertex $v$.

We claim that this yields a circular-arc representation of $G$.  Suppose
otherwise. There are two possibilities. Either the representation implies edges
that are not in $G$, or misses edges that are in $G$.  Clearly, all edges of
$G[X]$ are represented, by definition. Other edges are handled as follows.

First, consider $u,x\in X$ such that $\bar u$ exists.  Then, by construction,
the arcs for $\bar u$ and $x$ are disjoint if and only if the arc for $x$ is
contained in the arc for $u$.  Since the representation is normalized with
respect to $G$,  this holds by \ref{enum:n1} if and only if $N[x]\subseteq
N[u]$.  So by \ref{enum:circ1-0}, this is if and only if $\bar ux\not\in E$.

Next, consider $u,v\in X$ such that $\bar u$, $\bar v$ exist.  By construction,
the arcs for $\bar u$ and $\bar v$ are disjoint if and only if the arcs
representing $u$ and $v$ cover the circle. By \ref{enum:n2}, this is if and only
if $uv$ is a 2-overlap edge of $G$. So by \ref{enum:circ1-0} and
\ref{enum:circ1-i}, this is if and only if $\bar u\bar v\not\in E$.

This completes the proof.
\end{proof}

\begin{lemma}
Let $G$ be a graph with no universal vertices or true twins and let $G'$ be a
circular completion of $G$. Then:
\begin{compactenum}[\quad\rm (\thetheorem.1)]
\item $G'$ has no universal vertices or true twins,
\setreflabel{\rm (\thetheorem.\theenumi)}\label{enum:compl-nounivtt}
\item at least one vertex in every circular pair of $G'$ belongs to $V(G)$.
\setreflabel{\rm (\thetheorem.\theenumi)}\label{enum:compl-minimal}
\item $|V(G')|\geq 2|V(G)|-|S|$, where $S$ is the set of circularly-paired
vertices in $G$.
\setreflabel{\rm (\thetheorem.\theenumi)}\label{enum:compl-cardinality}
\addtocounter{claim}{\theenumi}
\end{compactenum}
\end{lemma}

\begin{proof}
First, we shall prove the following claim.

\begin{claim}\label{clm:delX}
Let $X\subseteq V(G')\setminus V(G)$ such that $X$ is nonempty. Then $G'-X$
is not circularly-paired.
\end{claim}

\begin{proofclaim}
We claim that every edge in $G'-X$ that also belongs to $G$ has the same type in $G'-X$ as it has in $G$. This can be seen as follows. Let $xy\in E(G)$. If $xy$ is an overlap edge in $G$, then clearly, it is also an overlap edge in any supergraph of $G$ which contains $G$ as an induced subgraph. Therefore, it is also an overlap edge in $G'-X$. Next, suppose that $xy$ is an inclusion edge in $G$ with $N_G[x]\subseteq N_G[y]$. As $G$ has no true twins, we have $N_G[x]\subset N_G[y]$. By definition of circular completion, $xy$ is an inclusion edge in $G'$, which can only mean that $N_{G'}[x]\subset N_{G'}[y]$. These two facts together imply that $N_{G'-X}[x]\subset N_{G'-X}[y]$. Therefore, $xy$ is an inclusion edge in $G'-X$ too. Finally, suppose $xy$ is a 2-overlap edge in $G$, or in other words, $xy$ is an overlap edge and $\{x,y\}$ a spanning pair in $G$. By what we showed above, it follows that $xy$ is an overlap edge in $G'-X$. If $\{x,y\}$ is not a spanning pair in $G'-X$, then in $G'-X$, there exists some vertex $z$ that is nonadjacent to one of $x,y$ and overlaps the other. As $G'-X$ is an induced subgraph of $G'$, this would mean that $z$ is nonadjacent to one of $x,y$ and overlaps the other in $G'$ too. Then clearly, $\{x,y\}$ is not a spanning pair in $G'$, contradicting the fact that $xy$ is a 2-overlap edge of $G'$ (recall that since $G'$ is a circular completion of $G$, the edge $xy$ has the same type in $G$ and $G'$). Therefore, $\{x,y\}$ is a spanning pair of $G'-X$, implying that $xy$ is a 2-overlap edge of $G'-X$. This shows that every edge of $G'-X$ that is also in $G$ has the same type in both the graphs. Now, since $G'-X$ is a smaller graph than $G'$, it must be the case that $G'-X$ is not circularly-paired.
\end{proofclaim}

The following claim is easy to see.
\begin{claim}\label{clm:circpairindsubgraph}
If $H$ is any graph and $\{x,y\}$ is a circular pair in $H$, then $\{x,y\}$ is a circular pair in any induced subgraph of $H$ that contains both $x$ and $y$.
\end{claim}

Now, we shall prove \ref{enum:compl-nounivtt}. Clearly, since $G'$ is a circularly-paired graph, there are no universal vertices in $G'$, as every vertex is nonadjacent to a vertex with which it forms a circular pair. Suppose that $x,y\in V(G')$ are true twins. Clearly, both $x$ and $y$ cannot be in $V(G)$ as if that were the case, they would be true twins in $G$ as well. So we shall assume without loss of generality that $y\in V(G')\setminus V(G)$. Now, consider the graph $G'-y$. By \ref{clm:delX}, we know that there exists a vertex $u$ that is not circularly-paired in $G'-y$. As $G'$ is circularly-paired, there exists some vertex $v$ such that $\{u,v\}$ is a circular pair in $G'$. If $v\neq y$, then by \ref{clm:circpairindsubgraph}, $\{u,v\}$ is also a circular pair in $G'-y$, contradicting the fact that $u$ is not circularly-paired in $G'-y$. Thus, we can conclude that $v=y$. But as $x$ and $y$ are true twins, the fact that $\{y,u\}$ is a circular pair in $G'$ implies that $\{x,u\}$ is also a circular pair in $G'$. This means, by \ref{clm:circpairindsubgraph}, that $\{x,u\}$ is also a circular pair in $G'-y$, again contradicting the fact that $u$ is not circularly-paired in $G'-y$. This proves \ref{enum:compl-nounivtt}.

Next, we shall prove \ref{enum:compl-minimal}. Suppose that there is a circular pair $\{u,v\}$ in $G'$ such that $u,v\in V(G')\setminus V(G)$. By \ref{clm:delX}, $G'-\{u,v\}$ is not circularly-paired. Let $x$ be a vertex in $G'-\{u,v\}$ that is not circularly-paired. Since $G'$ is circularly-paired, we know that there exists some vertex $y$ such that $\{x,y\}$ is a circular pair in $G'$. Then, by \ref{clm:circpairindsubgraph}, it must be the case that $y\in\{u,v\}$. We can assume without loss of generality that $y=u$, i.e. $\{x,u\}$ is a circular pair in $G'$. But now, $\{x,u\}$ and $\{u,v\}$ are both circular pairs in $G'$, which, by \ref{enum:compl-nounivtt}, is a graph without universal vertices or true twins. Since $u$, $v$ and $x$ are distinct vertices of $G'$, this contradicts \ref{enum:circ1-iii} and completes the proof of \ref{enum:compl-minimal}.

Finally, we prove \ref{enum:compl-cardinality}. Since by \ref{enum:compl-nounivtt}, $G'$ is a graph with no universal vertices or true twins, we know by \ref{enum:circ1-iii} that for every vertex $u\in V(G')$, there is a unique vertex $\bar u\in V(G')$ such that $\{u,\bar u\}$ is a circular pair in $G'$. Therefore, there are exactly $|V(G')|/2$ circular pairs in $G'$ and they form a partition of $V(G')$. Out of these, at most $|S|/2$ pairs consist only of vertices from $V(G)$, by \ref{clm:circpairindsubgraph}. By \ref{enum:compl-minimal}, exactly one vertex from the remaining at least $(|V(G')|-|S|)/2$ pairs is from $V(G')\setminus V(G)$. Therefore, we can conclude that $|V(G')|-|V(G)|\geq (|V(G')|-|S|)/2$ which gives $|V(G')|\geq 2|V(G)|-|S|$. This proves \ref{enum:compl-cardinality}.
\end{proof}

\begin{lemma}\label{lem:complexistence}
Let $G$ be a graph with no universal vertices or true twins, and let $S$
denote the set of circularly-paired vertices of $G$. Then $G$ has a circular
completion $G'$ with $|V(G')|=2|V(G)|-|S|$.
\end{lemma}

\begin{proof}
Let $G=(V,E)$.
For every $u\in S$ write $\bar u$ for the unique vertex $v$ such that
$\{u,v\}$ is a circular pair. This is guaranteed by Lemma \ref{lem:circ1}, since
$G$ has no true twins. Note that $\bar u\in S$. 

We prove by induction on $|V(G)\setminus S|$ that there exists a circular
completion $G'$ of $G$ such that $|V(G')|=2|V(G)|-|S|$. If $V\setminus S=
\emptyset$, then $G$ is circularly-paired, and the claim holds for $G'=G$.

We may therefore assume that $V\setminus S\neq\emptyset$ and let $v\in
V\setminus S$.  Add a new vertex $\bar v$ to $G$ whose neighbours are exactly
those vertices $u\in V(G)$ such that $N_G[u]\not\subseteq N_G[v]$. Call the
resulting graph $G^+$. In the following, for a vertex $u\in V(G^+)$, we shall
write $N[u]$ to denote $N_{G^+}[u]$.

\begin{claim}\label{clm:circ2-1}
All vertices in $S\cup\{v,\bar v\}$ are circularly-paired in $G^+$.
\end{claim}

\begin{proofclaim}
It is easy to see that $\{v,\bar v\}$ is a circular pair in $G^+$. Indeed, if
there exists a vertex $w$ in $G^+$ that is nonadjacent to $\bar v$, then by our
construction, we know that $N_G[w]\subseteq N_G[v]$ and because $\bar v$, the
only vertex in $V(G^+)\setminus V(G)$, is nonadjacent to $w$, it follows that
$N[w]\subseteq N[v]$. On the other hand, if there exists a vertex $w\neq\bar v$
in $G^+$ that is nonadjacent to $v$, then for each vertex $u\in N[w]\setminus\{
\bar v\}$, we have $N_G[u]\not\subseteq N_G[v]$ (as $w\in N_G[u]
\setminus N_G[v]$), implying by our construction that $u\bar v\in E(G^+)$.
This shows that $N[w]\subseteq N[\bar v]$. Therefore, $\{v,\bar v\}$ is a
circular pair in $G^+$.

Suppose that there
exists $u,\bar u\in S$ such that $\{u,\bar u\}$ is not a circular pair in $G^+$.
By symmetry between $u$ and $\bar u$, we can conclude that there exists $x\in
V(G^+)$, such that $ux\not\in E(G^+)$ and $N[x]\not\subseteq N[\bar u]$.

Suppose first that $x\neq\bar v$. Then clearly, $ux\not\in E(G)$ and because $\{
u,\bar u\}$ was a circular pair in $G$, we have $N_G[x]\subseteq N_G[\bar u]$.
Since $N[x]\not\subseteq N[\bar u]$, this can only mean
that $\bar v\in N[x]\setminus N[\bar u]$, or in other words, $x\bar v\in E(G^+)$
and $\bar u\bar v\not\in E(G^+)$. By our construction, the fact that $\bar u\bar
v\not\in E(G^+)$ implies that $N_G[\bar u]\subseteq N_G[v]$. This, together with
the earlier observation that $N_G[x]\subseteq N_G[\bar u]$ gives us $N_G[x]
\subseteq N_G[v]$. But then, by our construction, $x\bar v$ should not be an
edge in $G^+$, which is a contradiction.

Now suppose that $x=\bar v$. Then, we have $u\bar v\not\in E(G^+)$ and $N[\bar
v]\not\subseteq N[\bar u]$.

If $\bar u\bar v\not\in E(G^+)$, then consider a vertex $x\in V(G)$ such that
$vx\not\in E(G)$. Such a vertex must exist, since there are no universal
vertices in $G$. Then, recalling that $\{v,\bar v\}$ is a circular pair in
$G^+$, we have $N[x]\subseteq N[\bar v]$, implying that $xu,x\bar
u\not\in E(G^+)$. This also implies that $u\neq x\neq\bar u$. Therefore, we also
have $xu,x\bar u\not\in E(G)$, but this contradicts the fact that $\{u,\bar u\}$
is a circular pair in $G$.

Therefore, $\bar u\bar v\in E(G^+)$. Since $N[\bar v]\not\subseteq N[\bar u]$,
there exists $y\in N[\bar v]\setminus N[\bar u]$. Note that $y\neq\bar v$ and
since $y\bar u\not\in E(G^+)$, we also have $y\bar u\not\in E(G)$. This implies
that $N_G[y]\subseteq N_G[u]$ and, in particular, $uy\in
E(G)\subseteq E(G^+)$. Now, as $\bar vy\in E(G^+)$, it should be the case, by
our construction, that $N_G[y]\not\subseteq N_G[v]$. Therefore, there exists
some $x\in N_G[y]\setminus N_G[v]$. This clearly means that $xv\not\in E(G^+)$,
and since $\{v,\bar v\}$ is a circular pair of $G^+$, we have $N[x]\subseteq
N[\bar v]$.  Since $N_G[y]\subseteq
N_G[u]$, we have $xu\in E(G)$ and therefore $xu\in E(G^+)$. From that, we deduce
$u\bar v\in E(G^+)$, because $N[x]\subseteq N[\bar v]$.  But now this
contradicts our assumption that $u\bar v\not\in E(G^+)$.
\end{proofclaim}

\begin{claim}\label{clm:circ2-2}
Every edge $xy\in E(G)$ has the same type both in $G$ and in $G^+$.
\end{claim}

\begin{proofclaim}
Clearly, if $xy\not\in E(G)$, then $xy\not\in E(G^+)$ by construction.
Similarly, it is easy to see that if $xy$ is a 1-overlap edge in $G$, then it is
also a 1-overlap edge in $G^+$, since the vertices that make it fail to be an
inclusion or a 2-overlap edge in $G$, also make it fail in $G^+$.  By the same
token, no pair $\{x,y\}$ that is not spanning in $G$ can become spanning in
$G^+$.

So suppose that $xy$ is an inclusion edge where $N_G[x]\subseteq N_G[y]$.  For
contradiction, assume $N[x]\not\subseteq N[y]$. Then it must be the case that
$\bar v\in N[x]\setminus N[y]$. This implies, by the construction of $G^+$, that
$N_G[y]\subseteq N_G[v]$. Thus $N_G[x]\subseteq N_G[y]\subseteq N_G[v]$ which
implies that $x\bar v\not\in E(G^+)$ by construction, a contradiction.  This
proves that $xy$ remains the same inclusion edge in $G^+$.

Next, suppose that $\{x,y\}$ is a spanning pair in $G$ but not in $G^+$. Since
every inclusion edge in $G$ is also an inclusion edge in $G^+$ (as we just
showed), it follows that one of \ref{enum:c1}, \ref{enum:c2} fails for $\{x,y\}$
because of $\bar v$. By symmetry, $\bar v\not\in N[y]$ and $N[\bar
v]\not\subseteq N[x]$.  We have $N_G[y]\subseteq N_G[v]$ by construction, since
$y\bar v\not\in E(G^+)$. Assume first that $x\bar v\not\in E(G^+)$. Then
$N_G[x]\subseteq N_G[v]$ by construction.  Since $\{x,y\}$ is a spanning pair in
$G$, it follows by \ref{enum:c1} and \ref{enum:c2} that $V(G)=N[x]\cup N[y]$.
Thus $V(G)=N_G[x]\cup N_G[y]\subseteq N_G[v]$ which means that $v$ is a
universal vertex of $G$, a contradiction.  This shows that $x\bar v\in E(G^+)$
and thus there exists a different vertex $z\in N[\bar v]\setminus N[x]$. Note
that $z\in V(G)$.  Thus $z\not\in N_G[x]$ and so $N_G[z]\subseteq N_G[y]$ by
\ref{enum:c2}, since $\{x,y\}$ is a spanning pair of $G$. Recall that
$N_G[y]\subseteq N_G[v]$.  Thus $N_G[z]\subseteq N_G[y]\subseteq N_G[v]$, but
this implies $z\bar v\not\in E(G^+)$ by construction, a contradiction.  This
proves that $\{x,y\}$ remains a spanning pair in $G^+$.

In particular, if $xy\not\in E(G)$, then $\{x,y\}$ remains a circular pair,
while if $xy\in E(G)$, then $xy$ remains a 2-overlap edge.
\end{proofclaim}

\begin{claim}\label{clm:circ2-3}
$G^+$ contains no universal vertices or true twins.
\end{claim}

\begin{proofclaim}
The vertex $\bar v$ is non-adjacent to $v$, by construction.  Since $G$ contains
no universal vertex, every other vertex $x\in V(G)$ has at least one
non-neighbour in $G$, and thus in $G^+$. 

Similarly, if $G^+$ contains true twins $x,y\in V(G)$, then $N[x]=N[y]$ which
implies $N_G[x]=N_G[y]$ (as $G$ is an induced subgraph of $G^+$) and so $x,y$
are true twins in $G$. If $x\in V(G)$ is a true twin with $\bar v$, then $\{v,x
\}$ is a circular pair in $G^+$. This implies, by \ref{clm:circpairindsubgraph},
that $\{v,x\}$ is a circular pair in $G$. However, $v\not\in S$, a contradiction.
\end{proofclaim}

The claims \ref{clm:circ2-1} and \ref{clm:circ2-3} now allow us to apply the
inductive hypothesis to $G^+$. This yields a circularly-paired graph $G'$ such
that $G^{+}$ is an induced subgraph of $G'$ and every edge in $G'$ that is also
in $G^{+}$ has the same type in both $G^{+}$ and $G'$. Clearly, $G$ is an
induced subgraph of $G'$. It follows from \ref{clm:circ2-2} and the inductive
hypothesis that every edge that is in both $G$ and $G'$ has the same type in
$G'$ as it has in $G$.
By the inductive hypothesis, we have $V(G')=2|V(G^+)|-|S(G^+)|$ where $S(G^{+})$
denotes the set of circularly-paired vertices in $G^{+}$. From
\ref{clm:circ2-1}, it follows that $S\cup\{v,\bar v\}\subseteq S(G^{+})$. From
\ref{clm:circpairindsubgraph} and \ref{enum:circ1-iii}, it now follows that
$S(G^+)=S\cup\{v,\bar v\}$. As $V(G^+)=V(G)\cup\{\bar v\}$, we have $V(G')=
2|V(G)|-|S|$.
From \ref{enum:compl-cardinality}, we can conclude that $G'$ is a circular
completion of $G$.
\end{proof}

\begin{corollary}\label{cor:circpairsincompl}
Let $G$ be a graph with no universal vertices or true twins and let
$\{u,\bar u\}$ be a circular pair in $G$. Let $G'$ be a circular completion of
$G$. Then $\{u,\bar u\}$ is also a circular pair in $G'$.
\end{corollary}

\begin{proof}
Let $S$ be the set of circularly-paired vertices of $G$.
Let $S'\subseteq V(G')$ denote the set $\{u~|~u\in V(G)$ and there exists
$v\in V(G)$ such that $\{u,v\}$ is circular pair in	 $G'\}$. By
\ref{clm:circpairindsubgraph}, $S'\subseteq S$. As observed before, the
circular pairs of $G'$ form a partition of $V(G')$. Thus, there $|S'|/2$
circular pairs in $G'$ that contain only vertices from $V(G)$ and by
\ref{enum:compl-minimal}, there is exactly one vertex from $V(G')\setminus V(G)$
in the remaining $(|V(G')|-|S'|)/2$ circular pairs in $G'$. This tells us that
$|V(G')|-|V(G)|=(|V(G')|-|S'|)/2$ which gives $|S'|=2|V(G)|-|V(G')|$. From
Lemma~\ref{lem:complexistence} and \ref{enum:compl-cardinality}, we can
infer that $|V(G')|=2|V(G)|-|S|$. Combining the two equations gives us $|S'|=
|S|$, which means that $S'=S$. Now, let $\{u,\bar u\}$ be a circular pair in
$G$. Clearly, $u,\bar u\in S$ and therefore $u,\bar u\in S'$.
Suppose that $\{u,\bar u\}$ is not a
circular pair in $G'$. Since $G'$ is circularly-paired, there exists some $v\neq
\bar u$ such that $\{u,v\}$ is a circular pair in $G'$. As $u\in S'$, by
definition of $S'$, we have $v\in S'$ and therefore $v\in V(G)$. By
\ref{clm:circpairindsubgraph}, $\{u,v\}$ is a circular pair in $G$. Since $\{u,
\bar u\}$ is also a circular pair in $G$ and $u$, $\bar u$ and $v$ are distinct
vertices, this contradicts \ref{enum:circ1-iii}.
\end{proof}

\begin{lemma}
Any graph $G$ with no universal vertices or true twins has a unique (up to isomorphism) circular completion.
\end{lemma}

\begin{proof}
By Lemma~\ref{lem:complexistence}, we know that there exists a circular completion $G'$ of $G$.
Let $G''$ be another circular completion of $G$. Note that $G$ is an induced subgraph of both
$G'$ and $G''$. We define the bijection $f:V(G')\rightarrow V(G'')$ as:
\vspace{.1in}

\noindent\begin{center}
\begin{tabular}{p{.1\textwidth}p{.75\textwidth}}
$f(u)=u$,&if $u\in V(G)$\\
$f(u)=v$,& if $u\in V(G')\setminus V(G)$, where $v\in V(G'')\setminus V(G)$ is a vertex such that there exists $x\in V(G)$ for which $\{u,x\}$ is a circular pair in $G'$ and $\{v,x\}$ is a circular pair in $G''$.
\end{tabular}
\end{center}
\vspace{.1in}

Observe that \ref{enum:compl-minimal}, \ref{clm:circpairindsubgraph} and Corollary~\ref{cor:circpairsincompl} together ensures that $f$ is well-defined.
We claim that $f$ is an isomorphism between the graphs $G'$ and $G''$ and that therefore $G'$ and $G''$ are isomorphic graphs.
Note that by symmetry between $G'$ and $G''$, we only need to prove that if $uv\in E(G')$ then $f(u)f(v)\in E(G'')$.

Suppose that $uv\in E(G')$ and $f(u)f(v)\not\in E(G'')$.
If $u,v\in V(G)$, then $f(u)=u$ and $f(v)=v$. By the definition of circular completion, $G$ is an induced subgraph of both $G'$ and $G''$. Therefore, $uv\in E(G')$ implies that $f(u)f(v)\in E(G'')$, a contradiction.
Now suppose that $u\in V(G)$ and $v\in V(G')\setminus V(G)$. By definition of $f$, we have $f(u)=u$ and therefore $uf(v)\not\in E(G'')$. Let $v'\in V(G')$ be such that $\{v',v\}$ is a circular pair in $G'$. By \ref{enum:compl-minimal}, we know that $v'\in V(G)$ and by the construction of $f$, we know that $\{v',f(v)\}$ is a circular pair in $G''$. Now, $uf(v)\not\in E(G'')$ implies that $N_{G''}[u]\subseteq N_{G''}[v']$. Since $u,v'\in V(G)$ and $G$ is an induced subgraph of $G''$, we have $N_G[u]\subseteq N_G[v']$. As there are no true twins in $G$, we have $N_G[u]\subset N_G[v']$. Since $uv'$ is an inclusion edge in $G'$ (by definition of circular completion), this means that $N_{G'}[u]\subset N_{G'}[v']$. But $v\in N_{G'}[u]\setminus N_{G'}[v']$, a contradiction.

Finally, suppose that $u,v\in V(G')\setminus V(G)$. By \ref{enum:compl-minimal}, we know that there exist vertices $u',v'\in V(G)$ such that $\{u,u'\}$, $\{v,v'\}$ are circular pairs in $G'$ and by definition of $f$, we know that $\{f(u),u'\}$, $\{f(v),v'\}$ are circular pairs in $G''$. Since $f(u)f(v)\not\in E(G'')$, we have $N_{G''}[f(u)]\subseteq N_{G''}[v']$ and $N_{G''}[f(v)]\subseteq N_{G''}[u']$. From \ref{enum:circ1-i}, we can infer that $u'v'$ is a 2-overlap edge in $G''$. By definition of circular completion, $u'v'$ is a 2-overlap edge of $G'$ as well. But notice that we have $uu'\not\in E(G')$ and $v\in N_{G'}[u]\setminus N_{G'}[v']$. As $\{u',v'\}$ is a spanning pair of $G'$, this contradicts \ref{enum:circ1-0}.

This completes the proof.
\end{proof}

\begin{lemma}\label{lem:circ2}
Let $G$ be a graph with no universal vertices or true twins and let $G'$ be
its circular completion. Then, $G'$ is a circular-arc graph if and only if $G$
is.
\end{lemma}

\begin{proof}
Clearly, if $G'$ has a circular-arc representation, then so does $G$, since it
is induced in $G'$. For the converse, assume that $G$ is a circular-arc graph.
Recall that $G$ has no universal vertices or true twins. So by
Theorem~\ref{thm:hsu}, there exists a normalized circular-arc representation of
$G$.  By \ref{clm:circ2-2}, recall that every edge $uv\in E(G)$ has the same
type both in $G$ and $G'$. This implies that the representation of $G$ is
normalized with respect to $G'$.  Also, by \ref{enum:compl-minimal},
all vertices in $V(G')\setminus V(G)$
are circularly-paired with vertices in $V(G)$.  Therefore, by Lemma
\ref{lem:normalized}, $G'$ is a circular-arc graph.
\end{proof}

\section{Edge-labelled graphs}\label{sec:edgelabelledgraphs}
In this section, we deal with \emph{edge-labelled graphs}. An edge-labelled
graph $G$ is one in which every edge is labelled either as an {\em inclusion
edge} or an {\em overlap edge}.
Every vertex $u$ also has a loop $uu$, which is labelled as {\em inclusion edge}.
Note that the labels on the edges of $G$ need not correspond to the relationship
between the neighbourhoods of the two endpoints. That is, it is possible that
an edge $uv$ is labelled as an overlap edge even though $N[u]$ and $N[v]$ are
comparable and conversely, an edge $uv$ could be labelled as an inclusion edge
even though $N[u]$ and $N[v]$ are incomparable.

We shall now introduce some notions and terminology from \cite{mcc03},
which will be useful for us to prove Lemma~\ref{lem:non-inverting}.

Let $G=(V,E)$ be an edge-labelled graph. For $u,v\in V$,
we say that $u$ {\em overlaps} $v$, if $uv$ is an overlap edge. A vertex
$z\in V$ and a walk $P$ in $G$ {\em avoid} each other if every neighbour of $z$
on $P$ overlaps $z$, and whenever $z$ overlaps consecutive vertices $x,y$ on
$P$, then $xy$ is an inclusion edge.  For brevity, an edge $xy$ avoids $z$, if
the walk $x\tp y$ avoids $z$. We write $(x,z)\Delta (y,z)$ and $(z,x)\Delta (z,
y)$ whenever an edge $xy$ avoids $z$. Clearly, whenever this happens, neither
$xz$ nor $yz$ is an inclusion edge.

We say that a walk $P=x_1\tp x_2\tp\cdots\tp x_k$ in $G$ avoids a walk $Q=y_1\tp
y_2\tp\cdots\tp y_k$ if for every $i$, the vertex $x_i$ avoids the edge
$y_iy_{i+1}$, and the vertex $y_{i+1}$ avoids the edge $x_ix_{i+1}$. Note that
we may assume that for each $i$, either $x_i=x_{i+1}$ or $y_i=y_{i+1}$.
(If not, we can repeat some vertices on $P$ or $Q$ to achieve this.) Therefore,
whenever we have two such paths $P$ and $Q$, we shall assume that $(x_1,y_1)
\Delta (x_2,y_2)\Delta\cdots\Delta (x_k,y_k)$. This also implies that if $P$
avoids $Q$, then $Q$ avoids $P$, and therefore we shall simply say that $P$ and
$Q$ avoid each other.

A pair $(a,b)$ is {\em related} to pair $(u,v)$ if there exists an $au$-walk
and a $bv$-walk that avoid each other. As observed above, this means that there
exist $x_1,x_2,\ldots,
x_k,y_1,y_2,\ldots,y_k\in V$ such that $(a,b)\Delta (x_1,y_1)\Delta (x_2,y_2)$
$\Delta \cdots \Delta (x_k,y_k)\Delta (u,v)$. This definition
implies that $(a,b)$ is related to itself if and only if $(a,b)$ is not an
inclusion edge.

\subsection{\texorpdfstring{$\bf\Delta$}{}-implication classes}
\label{sec:deltaclasses}

A maximal set of pairwise related pairs is a {\em
$\Delta$-implication class}.  For a $\Delta$-implication class $A$, we write
$A^{-1}$ for the set $\{(u,v)~|~(v,u)\in A\}$. Observe that $(a,b)$ is related
to $(u,v)$ if and only if $(b,a)$ is related to $(v,u)$. Therefore, $A^{-1}$ is
also a $\Delta$-implication class. It is also easy to see that if $A\cap A^{-1}
\neq\emptyset$, then $A=A^{-1}$. We define $span(A)$ to be the set of all
vertices $u$ such that
$(u,v)\in A$ or $(v,u)\in A$ for some $v$.  In other words, $span(A)$ is the set
of all vertices incident to pairs in $A$. Clearly, $span(A)=span(A^{-1})$.

\begin{lemma}\label{lem:abc}
Let $A,B,C$ be $\Delta$-implication classes such that $C^{-1}\neq A\neq B$.
Suppose that there are vertices $a,b,c$ such that $(a,b)\in C$, $(b,c)\in A$,
$(a,c)\in B$.  Then\smallskip

\begin{compactenum}[\qquad \rm(\thetheorem.1)]
\item for every $(u,v)\in A$, we have $(a,u)\in C$ and $(a,v)\in
B$, and
\setreflabel{(\thetheorem.\theenumi)}\label{enum:first}
\item no pair in $A$ is incident to $a$.
\setreflabel{(\thetheorem.\theenumi)}\label{enum:second}
\addtocounter{claim}{\theenumi}
\end{compactenum}
\end{lemma}

\begin{proof}
Consider $(u,v)$ in $A$. Since $(b,c)$ is also in $A$, there exists a $bu$-walk
$P=x_1\tp x_2\tp\cdots\tp x_k$ and a $cv$-walk $Q=y_1\tp y_2\tp\cdots\tp y_k$
that avoid each other. Note that $(x_i,y_i)\in A$ for all $1\leq i\leq k$, by
the definition of a $\Delta$-implication class.
As observed before, we can assume that
$(x_1,y_1)\Delta(x_2,y_2)\Delta\cdots\Delta(x_k,y_k)$.

Recall that we assume $(a,b)\in C$ and $(a,c)\in B$. In other words, $(a,x_1)\in
C$ and $(a,y_1)\in B$. The statement \ref{enum:first} is a consequence of the
the following claim.

\begin{claim}
For every $i\in\{1,\ldots,k\}$, we have $(a,x_i)\in C$ and $(a,y_i)\in B$.
\end{claim}

\begin{proofclaim}
The proof is by induction. The claim holds for $i=1$ as noted above. 

Therefore assume that the claim holds for $i\geq 1$, namely $(a,x_i)\in C$,
$(a,y_i)\in B$. We prove that the claim holds for $i+1$. Note that the fact that
$(a,x_i)\in C$, $(a,y_i)\in B$ means that neither of $ax_i$, $ay_i$ is an
inclusion edge  (since each of $(a,x_i)$, $(a,y_i)$ belongs to some
$\Delta$-implication class).  Recall that $(x_i,y_i)\Delta(x_{i+1},y_{i+1})$. 

There are two possibilities: $x_i=x_{i+1}$ or $y_i=y_{i+1}$.  \smallskip

{\em Case 1:} suppose first that $x_i=x_{i+1}$ and $y_i\neq y_{i+1}$.  This
implies that neither of $x_iy_i$, $x_iy_{i+1}$ is an inclusion edge, while
$y_iy_{i+1}$ is an edge (inclusion or overlap).  Recall that neither of $ax_i$,
$ay_i$ is an inclusion edge. Note that $(a,x_{i+1})\in C$, since $(a,x_i)\in C$
and $x_i=x_{i+1}$.  Thus it remains to prove that $(a,y_{i+1})\in B$.  In
particular, we prove that $(a,y_i)\Delta(a,y_{i+1})$, which will imply
$(a,y_{i+1})\in B$, since $(a,y_i)\in B$. 

Recall that $y_iy_{i+1}$ is an edge (inclusion or overlap).  For contradiction,
suppose that $(a,y_i)\notdelta(a,y_{i+1})$. Since $y_iy_{i+1}$ is an edge, while
$ay_i$ is either a non-edge or an overlap edge, this means two possibilities:
either $ay_{i+1}$ is an inclusion edge, or $ay_i$, $ay_{i+1}$, $y_iy_{i+1}$ are
all overlap edges.

Suppose that $ay_{i+1}$ is an inclusion edge, and recall that
$(x_i,y_i)\Delta(x_{i+1},y_{i+1})$, and neither $ax_i$ nor $x_iy_{i+1}$ is 
an inclusion edge.  This implies that
$(a,x_i)\Delta(y_{i+1},x_i)=(y_{i+1},x_{i+1})\Delta(y_i,x_i)$.  And so since
$(a,x_i)\in C$, we conclude that $(y_i,x_i)\in C$. However, recall that
$(x_i,y_i)\in A$ and $A\neq C^{-1}$, a contradiction.

Suppose that $ay_i$, $ay_{i+1}$, $y_iy_{i+1}$ are overlap edges. Recall that
neither of $ax_i$, $x_iy_{i+1}$ is an inclusion edge.  If at least one of
$ax_i$, $x_iy_{i+1}$ is a non-edge, then we recall that
$(x_i,y_i)\Delta(x_{i+1},y_{i+1})$ and since $ay_{i+1}$ is an overlap edge, we
deduce that $(a,x_i)\Delta(y_{i+1},x_i)=(y_{i+1},x_{i+1})\Delta(y_i,x_i)$. And
so since $(a,x_i)\in C$, we again deduce that $(y_i,x_i)\in C$, but
$(x_i,y_i)\in A\neq C^{-1}$, a contradiction.  This implies that both $ax_i$,
$x_iy_{i+1}$ are overlap edges. Since $(x_i,y_i)\Delta(x_{i+1},y_{i+1})$ and
$y_iy_{i+1}$ is an overlap edge, it follows that $x_iy_i$ is a non-edge.  Thus
since $ay_i$ is an overlap edge, we deduce $(a,x_i)\Delta(y_i,x_i)$,  and so
again $(y_i,x_i)\in C$, but $(x_i,y_i)\in A\neq C^{-1}$.

This concludes Case 1.
\smallskip

\noindent{\em Case 2:} suppose that $y_i=y_{i+1}$ and $x_i\neq x_{i+1}$.  Apply
Case 1 with $C':=B$, and $B':=C$, and $A':=A^{-1}$. This is possible, since
$A'=A^{-1}\neq C=B'$ and $(C')^{-1}=B^{-1}\neq A^{-1}=A'$, because
$C^{-1}\neq A\neq B$.

\end{proofclaim}
\smallskip

Now for \ref{enum:second}, suppose that some edge $(u,v)\in A$ is incident to
$a$, i.e., $a\in \{u,v\}$. By (i), we deduce that $(a,u)\in C$ and $(a,v)\in B$.
In particular, $a\neq u$ and $a\neq v$, but $a\in\{u,v\}$, a contradiction.
\end{proof}

\begin{corollary}\label{cor:transitive}
Let $D$ be a $\Delta$-implication class such that $D\neq D^{-1}$. Then $D$ is
transitive.
\end{corollary}

\begin{proof}
Suppose that $(a,b),(b,c)\in D$ while
$(a,c)\not\in D$. As $(a,b),(b,c)\in D$, we know that $ab$ and $bc$ are not
inclusion edges. If $ac$ is an inclusion edge, then we have $(a,b)\Delta (c,b)$
which implies that $(c,b)\in D$. As we also have $(b,c)\in D$, this would mean
that $D=D^{-1}$, which is a contradiction to our assumptions. Therefore, $ac$
is not an inclusion edge and hence $(a,c)$ belongs to some $\Delta$-implication
class $B$. As we have assumed that $(a,c)\not\in D$, we have $B\neq D$. By
letting $A=C=D$ and applying Lemma~\ref{lem:abc}, we can conclude that no pair
from $A$ contains $a$, which contradicts the fact that $(a,b)$ is in $D=A$.
\end{proof}

\begin{corollary}\label{cor:nonedgeinspan}
Let $A$ be a $\Delta$-implication class such that for every $(u,v)\in A$, $uv$
is an overlap edge. Then, $span(A)$ induces a clique in $G$.
\end{corollary}

\begin{proof}
Suppose that there exist $x,y\in span(A)$ such that $xy$ is a non-edge. By
definition of $A$, we know that $(x,y)\not\in A\cup A^{-1}$. As $xy$ is a
non-edge, $(x,y)$ belongs to some $\Delta$-implication class $B$ such that $A
\neq B\neq A^{-1}$. As $x\in span (A)$, there exists some vertex $w$ such that
$(x,w)\in A$ or $(w,x)\in A$. Note that by definition of $A$, this means that
$xw$ is an overlap edge. If $wy$ is an inclusion edge, then we have $(x,w)\Delta
(x,y)$, implying that $(x,y)\in A\cup A^{-1}$, which is a contradiction.
Therefore, $wy$ is not an inclusion edge. Since we also know that $xy$ is a
non-edge and $xw$ is an overlap edge, it follows that $(x,y)\Delta (w,y)$ and
therefore, $(w,y)\in B$. Now, applying Lemma~\ref{lem:abc} with $A'=A$, $B'=
B^{-1}$ and $C'=B^{-1}$, we can conclude that no pair in $A$ is incident to $y$.
But this contradicts the fact that $y\in span(A)$. Therefore, there do not exist
$x,y\in span(A)$ such that $xy$ is a non-edge, implying that $span(A)$ induces a
clique in $G$.
\end{proof}

\subsection{\texorpdfstring{$\bf\Delta$}{}-modules}\label{sec:deltamodules}

A {\em $\Delta$-module} is a set $S\subseteq V$ that satisfies the following
two conditions.
\begin{compactenum}[\qquad(a)]
\item for every vertex $x\in V\setminus S$, one of the following is true:
\begin{compactitem}[--]
\item $x$ is nonadjacent to every vertex in $S$, or
\item for all $s\in S$, $xs$ is an inclusion edge, or
\item for all $s\in S$, $xs$ is an overlap edge,
\end{compactitem}
\item if $S$ is not a clique of $G$, then no vertex in $V\setminus S$
overlaps a vertex in $S$.
\end{compactenum}\smallskip

\noindent A $\Delta$-module $S$ is {\em trivial} if $S=V$ or $|S|=1$. Otherwise
it is {\em non-trivial}.
In \cite{mcc03}, the following is proved.

\begin{lemma}[\rm\cite{mcc03}]
If $A$ is a $\Delta$-implication class, then
\begin{compactenum}[\qquad\rm (\thetheorem.1)]
\item if $F$ is a $\Delta$-module containing vertices $a,b$, then $(a,b)\in A$
implies that $span(A)\subseteq F$.
\setreflabel{\rm (\thetheorem.\theenumi)}\label{enum:moduleoverlap}
\item $span(A)$ is a $\Delta$-module.
\setreflabel{\rm (\thetheorem.\theenumi)}\label{enum:spanmodule}
\end{compactenum}
\addtocounter{claim}{\theenumi}
\end{lemma}

\begin{proof}
For \ref{enum:moduleoverlap}, suppose that there exists $c\in span(A)\setminus
F$. As $c\in span(A)$, there exist $x_1,y_1,x_2,y_2,\ldots,x_k,y_k\in V(G)$ such
that $(a,b)=(x_1,y_1)\Delta (x_2,y_2)\Delta\cdots\Delta (x_k,y_k)$ where for
each $i$, $1\leq i\leq k-1$, either $x_i=x_{i+1}$ or $y_i=y_{i+1}$ and $c\in\{
x_k,y_k\}$. Clearly, as $c\not\in F$ and $x_1,y_1\in F$, there exist some $i$
for which $x_i,y_i\in F$ and some $z\not\in F$ such that
$x_{i+1}=z$ and $y_{i+1}=y_i$ or $x_{i+1}=x_i$ and
$y_{i+1}=z$. By symmetry, we shall assume that $x_{i+1}=z$ and $y_{i+1}=y_i$.
In short, we have $(x_i,y_i)\Delta (z,y_i)$. This implies that $x_iz$ is an
edge. As $F$ is a $\Delta$-module and $z\not\in F$, this means that $y_iz$ is
also an edge. As $(z,y_i)\in A$, we know that $y_iz$ is an overlap
edge and this by definition of $\Delta$-module means that $x_iz$ is also an
overlap edge. Again by definition of $\Delta$-module, this further implies that
$x_iy_i$ is an edge. But as $(x_i,y_i)\in A$, $x_iy_i$ is an overlap edge and
this contradicts the fact that $(x_i,y_i)\Delta (z,y_i)$.
This proves \ref{enum:moduleoverlap}.

\medskip

For proving \ref{enum:spanmodule}, we shall first prove the following simple
claim.

\begin{claim}\label{clm:straddle}
Suppose that $(x,y)\in A$, where $A$ is a $\Delta$-implication class and $z$
is a vertex such that one of the following conditions is true:
\begin{compactenum}[\qquad\rm (i)]
\item One of $xz,yz$ is an edge and the other is a non-edge,
\item One of $xz,yz$ is an inclusion edge and the other is an overlap edge,
\item Both $xz,yz$ are overlap edges and $xy$ is a non-edge.
\end{compactenum}
Then $z\in span(A)$.
\end{claim}

\begin{proofclaim}
Suppose that condition (i) is true. We shall assume by symmetry that $xz$ is an
edge and $yz$ is a non-edge. Then we have $(x,y)\Delta (z,y)$ and therefore $z
\in span(A)$. If condition (ii) is true, then again we can assume by symmetry
that $xz$ is an inclusion edge and $yz$ is an overlap edge, which implies that
$(x,y)\Delta (z,y)$, leading us to the conclusion that $z\in span(A)$. Finally,
if condition (iii) is true, then we again have $(x,y)\Delta (z,y)$, which again
gives us $z\in span(A)$.
\end{proofclaim}

Now we shall prove \ref{enum:spanmodule}. Suppose that $span(A)$ is not a
$\Delta$-module. Then there exist vertices $u,v\in span(A)$ and a vertex $z\not
\in span(A)$ such that one of the following situations occurs: (a) one
of $uz,vz$ is an edge and the other a non-edge or (b) one of $uz,vz$ is an
inclusion edge while the other is an overlap edge or (c) both $uz,vz$ are
overlap edges but $uv$ is not an edge. Since $u,v\in span(A)$, there exist
$x_1,y_1,x_2,y_2,\ldots,x_k,y_k\in span(A)$ such that $(x_1,y_1)\Delta
(x_2,y_2)\Delta\cdots\Delta(x_k,y_k)$, where each $(x_i,y_i)\in A$, $u\in
\{x_1,y_1\}$ and $v\in\{x_k,y_k\}$. As usual, we assume that for each $i$,
$1\leq i\leq k-1$, either $x_i=x_{i+1}$ or $y_i=y_{i+1}$. Clearly, if situation
(a) occurs, then there exists some $i$ such that one of $x_iz,y_iz$ is an edge
while the other is not. But this contradicts \ref{clm:straddle}. Similarly, if
situation (b) occurs, then there exists some $i$
such that one of $x_iz,y_iz$ is an inclusion edge while the other is an overlap
edge, which again contradicts \ref{clm:straddle}. Now, suppose that situations
(a) and (b) do not occur but situation (c) occurs.  As situations (a)
and (b) do not occur and $uz$ is an overlap edge,
it must be the case that for every vertex $w\in span (A)$,
$wz$ is an overlap edge. Therefore, if there is some $(x,y)\in A$
such that $xy$ is a non-edge, then it contradicts \ref{clm:straddle}. This tells
us that for any $(x,y)\in A$, $xy$ is an overlap edge.
By Corollary~\ref{cor:nonedgeinspan}, this means that
$span(A)$ induces a clique in $G$, contradicting the fact that $uv$ is a
non-edge. This proves \ref{enum:spanmodule}.
\end{proof}

\begin{lemma}\label{lem:deltaprime}
Suppose that every $\Delta$-module in $G$ is trivial. Suppose further that $G$
has at least one overlap edge, or at least one non-edge. Then either there are
exactly two $\Delta$-implication classes $A$, $A'$ where $A'=A^{-1}$ or one
$\Delta$-implication class $A=A^{-1}$.
\end{lemma}

\begin{proof}
Note that $G$ has at least one $\Delta$-implication class and also its reverse,
since every overlap edge (every non-edge) belongs to some $\Delta$-implication
class, and $G$ contains at least one such pair by our assumption.  For
contradiction, let $A$, $C$ be two distinct $\Delta$-implication classes, where
$C\neq A^{-1}$. By \ref{enum:spanmodule}, both $span(A)$ and $span(C)$ are
$\Delta$-modules. Since every $\Delta$-module of $G$ is trivial, we conclude
$span(A)=span(C)=V$.

Consider any $b\in V$. Since $b\in span(A)$, there exists $c\in V$ such that
$(b,c)\in A$ or $(c,b)\in A$. Likewise, since $b\in span(C)$, there exists $a\in
V$ such that $(a,b)\in C$ or $(b,a)\in C$.
By possibly replacing $A$ by $A^{-1}$ and $C$ by $C^{-1}$, we may assume that
$(b,c)\in A$ and $(a,b)\in C$.

Note that neither of $ab$, $bc$ is an inclusion edge, since the pairs $(a,b)$,
$(b,c)$ belong to $\Delta$-implication classes $C$ and $A$, respectively.  Since
$A\neq C\neq A^{-1}$, it follows that $(a,b)\notdelta(c,b)$.  This implies that
$ac$ is not an inclusion edge, since neither of $ab$, $cb$ is.  In particular,
$(a,c)$ belongs to some $\Delta$-implication class $B$. If $B\neq A$, then by
Lemma \ref{lem:abc} no pair in $A$ is incident to $a$, namely $a\not\in
span(A)$, but $a\in V=span(A)$, a contradiction.  Thus $B=A$ in which case we
apply Lemma \ref{lem:abc} with $A':=C$, and $B':=B$, and $C':=A^{-1}$. This is
possible since $(C')^{-1}=A\neq C=A'$ and $A'=C\neq A=B=B'$. From this we
conclude that no pair in $A'=C$ is incident to $c$, namely $c\not\in span(C)$,
but $c\in V=span(C)$, a contradiction.

This shows that no such $\Delta$-implication classes $A$, $C$ exist and so there
are only two $\Delta$-implication classes in $G$, namely some class $A$ and its
reverse $A^{-1}$. This concludes the proof.
\end{proof}

\subsection{Interval orderings, consistent representations and
\texorpdfstring{$\bf\Delta$-invertible pairs}{}}
%Let $G=(V,E)$ be an edge-labelled graph.

\begin{definition}[interval ordering]\label{def:intervalordering}
A linear ordering $<$ of $V(G)$ is said to be an \uline{interval ordering}
of $G$ if there are no three vertices $a,b,c\in V(G)$ with $a<b<c$ such that:
\begin{compactenum}[\qquad (i)]
\item $ab$ is a non-edge and $ac$ is an edge, or
\item $ab$ is an inclusion edge, $ac$ is a non-edge, and $bc$ is an edge, or
\item $ab$ is an overlap edge, $ac$ is an edge, and $bc$ is a non-edge, or
\item $ab$ and $bc$ are overlap edges, while $ac$ is an inclusion edge, or
\item $ab$ and $bc$ are inclusion edges, while $ac$ is an overlap edge.
\end{compactenum}
\end{definition}

\begin{definition}[consistent interval representation]
An interval representation of $G$ is said to be \emph{consistent} with the
labelling of $G$ if for any two vertices $u,v\in V(G)$, the interval for $u$
overlaps the interval for $v$ if and only if $uv$ is an overlap edge.
\end{definition}

\begin{lemma}\label{lem:orderrep}
$G$ has an interval ordering if and only if it has a consistent interval
representation.
\end{lemma}

\begin{proof}
It is easy to see that given a consistent interval representation of $G$, just
listing the vertices in the order of the left endpoints of their corresponding
intervals gives an interval ordering $<$ of $G$. For instance, the last property
holds because $a < b < c$ implies that the order of the left endpoints is
$l_a < l_b < l_c$ and hence $r_a > r_b$ since $ab$ is an inclusion edge,
and $r_b > r_c$ since $bc$ is an inclusion edge, whence $r_a > r_c$ and
so the interval for $a$ contains the interval for $c$, contradicting the fact
that $ac$ is an overlap edge.

Conversely, we shall show that if there exists an interval ordering $<$ of
$V(G)$, then $G$ has a consistent interval representation.
We shall prove by induction on $|V(G)|$ that there exists a consistent interval
representation of $G$ in which the left endpoints of the intervals occur in the
same order in which their corresponding vertices occur in the ordering $<$.
If $|V(G)|=1$, then assigning any interval to the only
vertex in $V(G)$ gives the required representation of $G$. Suppose that $|V(G)|
=n$ and that the claim is true for smaller values of $|V(G)|$. Let $x$ be the
least element of $<$ and let $G'$ be the edge-labelled graph obtained by
removing $x$ from $G$. Note that the labels of edges that are not incident on
$x$ are preserved in $G'$. Also, let $<'$ be the restriction of $<$ to $V(G')$.
By the inductive hypothesis, there exists a collection of intervals $\{[l_u,
r_u]\}_{u\in V(G')}$ that forms a consistent interval representation of $G'$ and
in which for $u,v\in V(G')$, we have $u<'v$ if and only if $l_u<l_v$.
We shall define the interval $[l_x,r_x]$ corresponding to $x$ that can be added
to this representation so that we will get a consistent interval representation
of $G$ in which for $u,v\in V(G)$, we will have $u<v$ if and only if $l_u<l_v$.

Define $l_x$ to be a value that is less than $\min\{l_u~|~u\in V(G')\}$.
Define $y$ to be the greatest vertex in the ordering $<$ such that $xy$ is an
edge. Notice that $y$ always exists as $xx$ is considered to be an edge.
Let $S$ be the set of vertices other than $x$ that have inclusion edges to $x$.
Let $t=\max(\{l_y\}\cup\{r_z~|~z\in S\})$. Define $r_x$ to be $t+\epsilon$,
where $\epsilon$ is chosen to be small enough to guarantee that no endpoint of
an interval of the representation occurs in the interval $(t,r_x)$.

Note that we only need to prove that for any $u\in V(G')$, the intervals
$[l_u,r_u]$ and $[l_x,r_x]$ are disjoint if and only if $xu$ is not an edge of
$G$ and overlap if and only if $xu$ is an overlap edge of $G$.

Let $u\in V(G')$.

First, let us consider the case when $xu$ is an edge of $G$. Clearly, we have
$y\geq u$, and therefore $r_x>l_u$. Since $l_x<l_u$, the intervals $[l_u,
r_u]$ and $[l_x,r_x]$ intersect.

Now, consider the case when $xu$ is not an edge of $G$. Suppose that
the intervals $[l_x,r_x]$ and $[l_u,r_u]$ intersect. As $l_x<l_u$, it must be
the case that $l_u<t$. If $l_u<l_y$, then we have $x<u<y$ and these three
vertices violate condition (i). Therefore, it must be the case that there exists
a vertex $z\in S$ such that $l_u<r_z$. If $l_u<l_z$, then we have $x<u<z$ and
these vertices violate condition (i). If $l_z<l_u$, then the intervals $[l_u,
r_u]$ and $[l_z,r_z]$ intersect, and therefore, by the inductive hypothesis,
$uz$ is an edge of $G$. Since $x<z<u$, these vertices violate condition (ii).
Therefore, the intervals $[l_u,r_u]$ and $[l_x,r_x]$ are disjoint.

Next, let us consider the case when $xu$ is an inclusion edge. We claim that the
interval $[l_x,r_x]$ contains the interval $[l_u,r_u]$. If that is not the case,
then we have $r_x<r_u$. This implies that $t<r_u$, but this is impossible by
definition of $t$ as $u\in S$.

Finally, consider the case when $xu$ is an overlap edge. Suppose that the
interval $[l_u,r_u]$ does not overlap the interval $[l_x,r_x]$. Then, it can
only mean that $l_x<l_u<r_u<r_x$. This implies that $r_u\leq t$. As $u\not\in
S$, we can conclude that $r_u<t$. This can happen only if either $r_u<l_y$ or
there exists a vertex $z\in S$ such that $r_u<r_z$. If $r_u<l_y$, then $uy$ is
not an edge (by the inductive hypothesis) and therefore, the vertices $x,u,y$
violate condition (iii). If there exists a vertex $z\in S$ such that $r_u<r_z$,
then we have two possibilities: either (a) $l_u<l_z$, in which case $uz$ is an
overlap edge (by the inductive hypothesis), implying that $x,u,z$ violate
condition (iv), or (b) $l_z<l_u$, in which case $uz$ is an inclusion edge (by
the inductive hypothesis), implying that $x,z,u$ violate condition (v).

Therefore, the intervals $\{[l_v,r_v]\}_{v\in V(G)}$ form a consistent interval
representation of $G$.
\end{proof}

\begin{definition}[$\Delta$-invertible pair]
A pair of vertices $\{a,b\}\subseteq V(G)$ is said to be a $\Delta$-invertible
pair in $G$ if $(a,b)$ is related to $(b,a)$ in $G$.
\end{definition}

Note that if $A$ is a $\Delta$-implication class of $G$ and $A=A^{-1}$, then
any pair from $A$ is a $\Delta$-invertible pair in $G$. Conversely, if $\{a,b\}$
is a $\Delta$-invertible pair in $G$, then the $\Delta$-implication class $A$
that contains $(a,b)$ has the property that $A=A^{-1}$.

\begin{definition}
A \emph{consistently edge-labelled graph} $G$ is an edge-labelled graph that
%said to have a \emph{consistent labelling}
%if it
has a transitive orientation $D_c$ of its inclusion edges such that for each
$(u,v)\in D_c$, $N[v]\subseteq N[u]$.
%satisfies the following two conditions:
%\begin{compactenum}[(i)]
%\item If $uv$ is labelled as an inclusion edge, then either $N[u]\subseteq
%N[v]$ or $N[v]\subseteq N[u]$.
%\item If $xy$ and $yz$ are inclusion edges such that $N[x]\subset N[y]
%\subset N[z]$, then $xz$ is also labelled as an inclusion edge.
%\end{compactenum}
\end{definition}

\begin{lemma}\label{lem:delinvpair}
Let $G$ be a consistently edge-labelled graph. Then $G$ has a consistent
interval representation if and only if $G$ contains no $\Delta$-invertible pair.
\end{lemma}

\begin{proof}
It follows from the definition of the $\Delta$ relation that if $(a,b)$ is
related to $(c,d)$ in $G$, then in any consistent
interval representation of $G$, the left endpoint of the interval of $a$ is
before the left endpoint of the interval of $b$ if and only if the left endpoint
of the interval of $c$ is before the left endpoint of the interval of $d$.
Therefore, if $G$ contains a $\Delta$-invertible pair, then it cannot have a
consistent interval representation.

Suppose that $G$ contains no $\Delta$-invertible pair. We shall show that $G$
has an interval ordering, which by Lemma~\ref{lem:orderrep} is enough to prove
that $G$ has a consistent interval representation.

We shall prove this by producing
an \uline{interval orientation} of $G$. This is an acyclic orientation $D$
of non-edges and overlap edges of $G$ such that for $x,y,z\in V(G)$, if $z$
avoids $xy$, then we either have $(x,z),(y,z)\in D$ or $(z,x),(z,y)\in D$.
In other words, it is an acyclic orientation of the non-edges and overlap edges
of $G$ such that if $(a,b)$ is related to $(c,d)$, where $a,b,c,d\in V(G)$,
then either both or neither of the two pairs are in the orientation.

Note that no $(x,y)$ is related to $(y,x)$, since $G$ contains no
$\Delta$-invertible pair.

We shall prove that $G$ has an interval orientation by induction on
$|V(G)|$. Let $S\subseteq V(G)$ be a $\Delta$-module of $G$ such that $S\neq
V(G)$. We can replace $S$
by any one vertex of $S$, construct an interval orientation of the result, and
then substitute back $S$ along with its interval orientation.  We claim that the
resulting orientation, which we shall call $D$, is an interval orientation of
$G$. Clearly, $D$ is acyclic.
Suppose there are $x,y,z\in V(G)$ such that $z$ avoids $xy$, but $(x,z)$,
$(z,y)$ are in $D$. Exactly two of $x,y,z$ must be in $S$.  The case $x,y\in S$
and $z\in V(G)\setminus S$ is impossible, since every vertex outside $S$ either
dominates $S$ or is dominated by $S$ in the orientation, by construction.  So by
symmetry, $x,z\in S$ and $y\in V(G)\setminus S$. Note that since $z$ avoids
$xy$, we have $(x,z)\Delta (y,z)$, implying that $(x,z)$ and $(y,z)$ are in the
same $\Delta$-implication class. But $S$ is a $\Delta$-module that contains $x$
and $z$ but not $y$, contradicting \ref{enum:moduleoverlap}. Therefore, $D$ is
an interval orientation of $G$.

Thus it remains to assume that $G$ contains no non-trivial $\Delta$-module.
If $G$ has no non-edges or overlap edges, then clearly it has an interval
orientation (which is the empty set). Otherwise, because $G$ contains no
$\Delta$-invertible pair, Lemma~\ref{lem:deltaprime} implies that there is
a $\Delta$-implication
class $A$ such that for every $ab$ that is an overlap edge or a non-edge of
$G$, exactly one of $(a,b)$ or $(b,a)$ is in $A$. By
Corollary~\ref{cor:transitive}, $A$ is acyclic.
As $A$ is a $\Delta$-implication class of $G$, we know that
for any two related pairs $(a,b)$, $(c,d)$ of $G$, one of them is in $A$ if
and only if the other also is. Therefore, $A$ is an interval orientation of
$G$.

This constructs an interval orientation $D$ of $G$. Suppose that $D$ is not
transitive. That is, there exist $a,b,c\in V(G)$ such that $(a,b),(b,c)\in D$
while $(a,c)\not\in D$. Clearly, $(c,a)\not\in D$ as $D$ is acyclic. Therefore
$ac$ is an inclusion edge, which implies that $(a,b)\Delta (c,b)$. But
this gives $(b,c)\Delta (c,b)$, which is impossible as $G$ has no
$\Delta$-invertible pair. Therefore, $D$ is transitive.
%Let $D_c$ be an
%orientation of those inclusion edges $uv$ of $G$ where $N[u]\neq N[v]$ obtained
%as follows: each such inclusion edge $uv$ is oriented from $u$ to $v$ if and
%only if $N[v]\subset N[u]$.
%and every inclusion edge $uv$ with $N[u]=N[v]$ in an arbitrary direction.
%By the definition of edge-labelled graphs,
%Since $G$ has a consistent labelling, $D_c$ is transitive.
As $G$ is a consistently edge-labelled graph, there is a transitive orientation
$D_c$ of the inclusion edges such that for each $(u,v)\in D_c$, we have $N[v]
\subseteq N[u]$.
% with both
% endpoints in $X$, from smaller to larger neighbourhoods. Let $D_t$ denote the
%reverse of $D_c$. 
It can be seen that $D\cup D_c$ is also a transitive orientation.
%This can be
%shown as follows. Suppose that $(a,b),(b,c)\in D\cup D_c$, but $(a,c)\not\in D
%\cup D_c$. Since $D$ and $D_c$ are disjoint and both of them are transitive,
%it must be the case that one
%of $(a,b),(b,c)$ belongs to $D$ and the other belongs to $D_c$. Let us assume
%that $(a,b)\in D$ and $(b,c)\in D_c$ (the other case is symmetric). Note that
%this gives us $N[c]\subset N[b]$. Since $(a,c)\not \in D\cup D_c$, it must be
%the case that $N[a]=N[c]$. But this also tells us that $N[a]\subset N[b]$,
%implying that $(b,a)\in D_c$, a contradiction.
Indeed, note that $D\cup D_c$ is a tournament and that $D$ and
$D_c$ are disjoint. So if
it is not transitive, it contains an directed triangle. At most one edge of this
triangle can be from each of $D$, $D_c$, since they are both transitive
orientations, a contradiction.

%By symmetry, $D\cup D_t$ is also transitive.  So,
%following \cite{mcc03},
%Let $L$ be a linear extension of $D\cup D_c$.
% and let $R$
%be the reverse of a linear extension of $D\cup D_t$.  As shown in \cite{mcc03},
%there exists an interval representation $\{[l_u,r_u]\}_{u\in X}$ of $\cal X$
%where the relative order of left- and right-endpoints is $L$ and $R$,
%respectively.
This shows that $D\cup D_c$ is a transitive tournament, and therefore, a linear
order.
It is easy to verify using Definition~\ref{def:intervalordering} that
$D\cup D_c$ is an interval ordering of $G$. This completes the proof.
\end{proof}

\section{Non-inverting sets}

% we shall be using both normal graphs and edge-labelled graphs.
%Whenever the graph we are referring to is edge-labelled, we shall explicitly
%say so.
%For normal graphs, we shall implicitly assume that the edges have
%labels corresponding to their edge types in $G$ as given in
%Definition~\ref{def:edgetypes}.

Let $G=(V,E)$ be a graph. In this section, the edges of $G$ are classified into
different edge types as given in Definition~\ref{def:edgetypes}.

\begin{definition}[non-inverting set]
A set $X\subseteq V$ is a \uline{non-inverting set} of $G$ if the following
holds:
\begin{compactenum}[\quad(V1)]
\item no 2-overlap edge of $G$ has both endpoints in $X$,
\setreflabel{\rm(V\theenumi)}\label{enum:m1}
\item no $x,y\in X$ are such that there is an $xy$-walk $P\subseteq X$ and a
$yx$-walk $Q\subseteq X$ that avoid each other.
\setreflabel{\rm(V\theenumi)}\label{enum:m2}
\end{compactenum}
\end{definition}

The next lemma constitutes the final ingredient needed for Theorem
\ref{thm:main}.

\begin{lemma}\label{lem:non-inverting}
Let $G$ be a circularly-paired graph with no true twins. If $G$ has a
non-inverting set that contains exactly one vertex of every circular pair of
$G$, then $G$ is a circular-arc graph.
\end{lemma}

\begin{proof}
Let $X$ be a non-inverting set that contains exactly one vertex of every
circular pair of $G$. 
%The subsequent proof is based on \cite{mcc03}.  We need
%to introduce some terminology (some of which is phrased in our language, but
%identical).
Let $\cal X$ be the subgraph induced by $X$ in $G$ in which every edge is
labelled with the edge type it has in $G$. We can get a transitive orientation
$D_c$ of the inclusion edges of $\cal X$ by orienting every such edge $uv$ from 
$u$ to $v$ if and only if $N_G[v]\subseteq N_G[u]$.
%Every inclusion edge $uv$ of $\cal X$
%is oriented from $u$ to $v$ if and only if $N[u]\supseteq N[v]$ in $G$.
Clearly, for $u,v\in X$, whenever we have $N_G[v]\subseteq N_G[u]$, we also have
$N_{\cal X}[v]\subseteq N_{\cal X}[u]$. Therefore, for each $(u,v)\in D_c$, we
have $N_{\cal X}[v]\subseteq N_{\cal X}[u]$. This shows that $\cal X$ is a
consistently edge-labelled graph.
%Let the relation $\Delta$ on pairs
%from $X$, $\Delta$-implication classes, $\Delta$-modules, ``avoids'' and
%``related to'' be defined on $\cal X$
%as given in Section~\ref{sec:delta}. 

Since \ref{enum:m2} holds for $X$ in $G$, we can conclude that $\cal X$ contains
no $\Delta$-invertible pair. Therefore,
from Lemma~\ref{lem:delinvpair}, we know that there is a consistent interval
representation $\{[l_u,v_u]\}_{u\in X}$ for $\cal X$. Since the labels on the
edges of $\cal X$ came from the neighbourhood relations from $G$, we can
conclude that if $u,v\in X$ and
$uv$ is an overlap edge of $G$, then intervals $[l_u,r_u]$, $[l_v,r_v]$ overlap
and if $uv$ is an inclusion edge of $G$ and $N_G[u]\subseteq N_G[v]$, then
$[l_v,r_v]$ properly contains $[l_u,r_u]$. Therefore, this is a circular-arc
representation of $G[X]$ normalized with respect to $G$. (For this recall that
no 2-overlap edge has both endpoints in $X$.) Thus by
Lemma~\ref{lem:normalized}, we conclude that $G$ is a circular-arc graph as
claimed.
\end{proof}

\section{Proof of Theorem {\ref{thm:main}}}\label{sec:charac}
\newcounter{save}
\setcounter{save}{\thetheorem}
\setcounter{theorem}{1}
\indent\ref{enum:I}$\Rightarrow$\ref{enum:II}:
Let $G=(V,E)$ be a circular-arc graph with no universal vertices or true
twins. By Theorem \ref{thm:hsu}, let $\{[l_u,r_u]\}_{u\in V}$ be a normalized
circular-arc representation of $G$.  Suppose that there exists an $xy$-walk
$P=x_1\tp x_2\tp\cdots\tp x_k$ and a $yx$-walk $Q=y_1\tp y_2\tp\cdots\tp y_k$
such that $P$ and $Q$ avoid each other and both walks avoid $z$. Recall that
we may assume that for all $i$, either $x_i=x_{i+1}$ or $y_i=y_{i+1}$.

Observe that the mirror image of the representation is
also normalized.  Thus by symmetry, we may assume that $l_z<l_{x_1}<l_{y_1}$.

\begin{claim}
For all $1\leq i\leq k$, we have $l_z<l_{x_i}<l_{y_i}$.
\end{claim}

\begin{proofclaim}
Consider the smallest $i$ for which the claim fails. We will derive a
contradiction. Clearly $i\geq 2$, since $l_z<l_{x_1}<l_{y_1}$. Thus we have
$l_z<l_{x_i}<l_{y_i}$ while $l_z<l_{y_{i+1}}<l_{x_{i+1}}$.

First suppose that $x_i=x_{i+1}$. Since $P$ and $Q$ avoid each other, the vertex
$x_i$ avoids
the edge $y_iy_{i+1}$. Also $z$ avoids $y_iy_{i+1}$, since it avoids both walks.
Since $x_i=x_{i+1}$, we have $l_z<l_{y_{i+1}}<l_{x_i}<l_{y_i}$ where
$y_iy_{i+1}\in E$. Thus by Lemma \ref{lem:avoid}, at least one of $z$, $x_i$
does not avoid $y_iy_{i+1}$, a contradiction.

Similarly if $y_i=y_{i+1}$. We find that $y_{i+1}$ and $z$ avoid the edge
$x_ix_{i+1}$ while $l_z<l_{x_i}<l_{y_{i+1}}<l_{x_{i+1}}$. This is again
impossible by Lemma \ref{lem:avoid}.
\end{proofclaim}

From this claim, we deduce $l_z<l_{x_k}<l_{y_k}$. Since $x_k=y_1$ and $y_k=x_1$,
this means that $l_z<l_{y_1}<l_{x_1}$. However, we assume that
$l_z<l_{x_1}<l_{y_1}$, a contradiction.

Therefore $G$ contains no such paths $P$, $Q$ and consequently contains no
anchored invertible pair. By Lemma \ref{lem:circ2}, the same holds for the
circular completion of $G$. This proves \ref{enum:II}.

\medskip

\ref{enum:II}$\Rightarrow$\ref{enum:I}:
Let $G=(V,E)$ be graph with no universal vertices or true twins.  By
Lemma~\ref{lem:circ2}, we may assume that $G$ is circularly-paired (otherwise we
replace $G$ by its circular completion).  Throughout the rest of the proof, we
shall assume that $G$ contains no anchored invertible pair.

Since $G$ is circularly-paired, then by Lemma \ref{lem:circ1}, every vertex
$u\in V$ is circularly-paired with a unique other vertex. We shall write $\bar
u$ to denote this vertex.  Note that $\bar{\,\bar u\,}=u$.

Throughout the rest of the proof, let $z$ denote a fixed vertex of $G$ of
minimum degree.  Let $X$ denote the set of vertices that overlap $z$.

\begin{claim}\label{clm:h0}
For each $x\in X$, the edge $xz$ is a
1-overlap edge
\end{claim}

\begin{proofclaim}
If $xz$ were a 2-overlap edge, then $N[\bar x]\subseteq N[z]$ by
\ref{enum:circ1-i} of Lemma \ref{lem:circ1}. Since there are no true twins, the
degree of $\bar x$ would be necessarily strictly smaller than that of $z$, a
contradiction.

\end{proofclaim}

For vertices $x,y\in X$, we say that $x$ and $y$ \uline{{\em disagree}}
if there exists an $x\bar z$-walk $P$ and a $\bar zy$-walk $Q$ such that $P$
and $Q$ avoid each other, and both $P$ and $Q$ avoid $z$.

\begin{claim}\label{clm:h1}
Let $x,y\in X$. If $xy\not\in E$, then $x$ and $y$ disagree.
\end{claim}

\begin{proofclaim}
By \ref{clm:h0}, $xz,yz$ are 1-overlap edges. Thus by \ref{enum:circ1-ii},
both $x\bar z, y\bar z$ are 1-overlap edges.
So if $xy\not\in E$, then $x$ and $y$ disagree as witnessed by the
walks $P=x\tp x\tp\bar z$ and $Q=\bar z\tp y\tp y$.
\end{proofclaim}

Let $H$ be the graph with vertex set $X$ where two vertices are adjacent if
they disagree.

\begin{claim}\label{clm:h2}
$H$ is a bipartite graph.
\end{claim}

\begin{proofclaim}
For contradiction, suppose that $H$ contains a cycle $x_1\tp x_2\tp\cdots\tp
x_k\tp x_{k+1}=x_1$ where $k$ is odd.  

By definition, for each $i\in\{1,\ldots,k\}$, there exist walks $P_i$, $Q_i$
where $P_i$ is a walk from $x_i$ to $\bar z$, where $Q_i$ is a walk from $\bar
z$ to $x_{i+1}$, and where $P_i$ and $Q_i$ avoid each other, and both walks
avoid $z$.  Then let $P$ be the concatenation of $P_1,Q_2,P_3,\ldots,Q_{k-1},
P_k$ and let $Q$ be the concatenation of $Q_1,P_2,Q_3,\ldots,P_{k-1},Q_k$. Then
$P$ is a walk from $x_1$ to $\bar z$, while $Q$ is a walk form $\bar z$ to
$x_{k+1}=x_1$.  It follows that $P$ and $Q$ avoid each other, since for each
$i$, $P_i$ and $Q_i$ avoid each other.  By the same token,
both $P$ and $Q$ avoid $z$, since each $P_i$ and $Q_i$ avoid $z$.  Therefore
$(x_1$, $\bar z)$ is an anchored invertible pair with anchor $z$,
a contradiction.
\end{proofclaim}

In view of \ref{clm:h2}, let $(Y,X\setminus Y)$ be a bipartition of $V(H)=X$
into two independent sets of $H$.  Define $Z$ to be the union of $V\setminus
N[z]$ and $Y$.  The following is an easy consequence of \ref{clm:h1}.

\begin{claim}\label{clm:h3}
No two vertices in $Z$ form a spanning pair.
\end{claim}

\begin{proofclaim}
Let $x,y\in Z$. Suppose that $\{x,y\}$ is a spanning pair.  If $xz\not\in E$,
then $N[z]\subseteq N[y]$ by \ref{enum:c2}. In particular, $y\in N[z]$ and so
$y\in Y$, since $y\in Z$ and $Y=N(z)\cap Z$. This, however, is impossible, since
every vertex in $X\supseteq Y$ overlaps $z$.  Therefore $xz\in E$
likewise $yz\in E$.  This implies $x,y\in Y$.

Now if $xy\not\in E$, then $x$ and $y$ disagree by \ref{clm:h1}, in which
case $x$, $y$ are adjacent in $H$. This is impossible, since $Y$ is an
independent set of $H$.

Thus $xy\in E$ and so $xy$ is a 2-overlap edge. This implies that $\bar x\bar
y\not\in E$ by \ref{enum:circ1-0} and \ref{enum:circ1-i} of Lemma
\ref{lem:circ1}.  Also $x\bar x,y\bar y\not\in E$ by definition. In addition,
$z\bar x$ and $z\bar y$ are 1-overlap edges by \ref{enum:circ1-ii}, since $zx$
and $zy$ are 1-overlap edges because $x,y\in X$ and \ref{clm:h0}.  This shows
that both $\bar x$ and $\bar y$ are in $X$. Thus by \ref{clm:h1}, $x\bar x$,
$\bar x\bar y$, $\bar y y$ are edges of $H$. It follows that at least one of
these edges has both endpoints in $Y$ or in $X\setminus Y$.  But both $Y$ and
$X\setminus Y$ are independent sets of $H$, a contradiction.

\end{proofclaim}

\begin{claim}\label{clm:h3.5}
For every $u\in V$, exactly one of $u$, $\bar u$ belongs to $Z$.
\end{claim}

\begin{proofclaim}
Both $u$, $\bar u$ cannot belong to $Z$, because $Z$ contains no spanning pair
by \ref{clm:h3}. If both $u$ and $\bar u$ do not belong to $Z$, then neither of
$u$, $\bar u$ overlaps or is non-adjacent to $z$. It follows
that both $uz$ and $\bar uz$ are inclusion edges. But this is impossible by
\ref{enum:circ1-0} and \ref{enum:circ1-i}.
\end{proofclaim}

\begin{claim}\label{clm:h4}
Let $a,b,c\in Z$.  If $a$ avoids $bc$ but $z$ does not, then $a\in V\setminus
N[z]$, $b,c\in Y$, and $b$ overlaps $c$.
\end{claim}

\begin{proofclaim}
Let $a,b,c\in Z$ be such that $a$ avoids $bc$ but $z$ does not.  This means that
$z$ is adjacent to one or both of $b,c$.  We observe further.  If $zb\in E$,
then $z$ overlaps $b$, since $Y=N[z]\cap Z$ and $Y\subseteq X$.  Likewise, if
$zc\in E$, then $z$ overlaps $c$.  Thus since $z$ does not avoid $bc$, it must
be the
case that $z$ overlaps both $b$ and $c$, and $b$ overlaps $c$.  It follows that
$b,c\in Y$, since $Y=N[z]\cap Z$.  On the other hand, since $a$ avoids $bc$, and
$b$ overlaps $c$, it follows that $a$ is non-adjacent to at least one of
$b$,$c$.

Suppose that $a\in Y$. Then we find by \ref{clm:h1} that $a$ disagrees with
at least one of $b$, $c$. Thus $a$ is adjacent in $H$ to at least one of $b$,
$c$. But this is impossible, since $a,b,c\in Y$ and $Y$ is an independent set of
$H$.  Therefore $a\in V\setminus N[z]$ as claimed.
\end{proofclaim}

\begin{claim}\label{clm:h5}
There are no $x,y\in Y$ such that there exists an $x\bar z$-walk
$P\subseteq Z$ and a $\bar zy$-walk $Q\subseteq Z$ that avoid each other.
\end{claim}

\begin{proofclaim}
For contradiction, consider $x,y\in Y$ be such that there exists an $x\bar
z$-walk $P=x_1\tp x_2\tp\cdots\tp x_k$ and a $\bar zy$-walk $Q=y_1\tp y_2\tp
\cdots\tp y_k$ where $P$ and $Q$ avoid each other and all vertices on both $P$
and $Q$ belong to $Z$.  Pick $x,y$ so that $P$, $Q$ are shortest possible and
(subject to that) contain as many loops as possible.

If both $P$ and $Q$ avoid $z$, then $x$ and $y$ disagree.
Thus $x,y$ are adjacent in $H$. This is impossible, since $x,y\in Y$ and $Y$ is
an independent set of $H$.

Thus $z$ does not avoid some edge on $P$ or $Q$. Note the symmetry between $P$
and $Q$. (We can swap $x,y$ and/or replace $P$, $Q$ by the reverse of $Q$ and
reverse of $P$, respectively.) 

Therefore, by this symmetry, we may assume that for some $i<k$, the edge $y_i
y_{i+1}$ of $Q$ does not avoid $z$. Since $P$ and $Q$ avoid each other, this
implies that
$x_i$ avoids $y_iy_{i+1}$.  This implies by \ref{clm:h4} that $y_i,y_{i+1}\in
Y$, that $y_i$ overlaps $y_{i+1}$, and that $x_i\in V\setminus N[z]$.

Now recall that $i\leq k-1$.  Suppose first that  $x_i=\bar z$. Then $P'=x_1\tp
x_2\tp\cdots\tp x_i$ is an $x_1\bar z$-walk and $Q'=y_1\tp y_2\tp\cdots\tp y_i$
is an $\bar zy_i$-walk such that $P'$ and $Q'$ avoid each other, both are walks
with all vertices in $Z$, and $x_1,y_i\in Y$.  But $P'$ and $Q'$ are shorter
than $P$ and $Q$, contradicting our choice of $P$ and $Q$.

We may therefore assume that $x_i\neq \bar z$.  Since $y_i\in X$, we have by
\ref{clm:h0} that $zy_i$ is a 1-overlap edge.  Thus $\bar zy_i$ is a 1-overlap
edge by \ref{enum:circ1-ii}.  Likewise, since $x_iz\not\in E$, we deduce $N[\bar
z]\supseteq N[x_i]$ by \ref{enum:circ1-0}.  This implies that $y_i$ avoids the
edge $x_i\bar z$.  Note $y_i\neq y_{i+1}$, since $y_i$ overlaps $y_{i+1}$.

Thus we let $P'=x_1\tp x_2\tp\cdots\tp x_i\tp\bar z$ and let $Q'=y_1\tp
y_2\tp\cdots\tp y_i\tp y_i$. It follows that $P'$ and $Q'$ avoid each other,
and the two walks
have all vertices in $Z$ where $P'$ is an $x_1\bar z$-walk and $Q'$ is a $\bar z
y_i$-walk. Moreover, we have $x_1,y_i\in Y$.  But since $i\leq k-1$, the two
walks are either shorter than $P$ and $Q$ (if $i\leq k-2$) or contains more
loops (since $y_i\neq y_{i+1}$), again
contradicting our choice of $P$ and $Q$.

This concludes the proof.
\end{proofclaim}

A proof similar to the above gives the following.

\begin{claim}\label{clm:h6}
$Z$ is a non-inverting set of $G$.
\end{claim}

\begin{proofclaim}
The condition \ref{enum:m1} holds by \ref{clm:h3}. For \ref{enum:m2}, assume for
contradiction that for some $x,y\in Z$ there is an $xy$-walk $P=x_1\tp x_2\tp
\cdots\tp x_k$ and a $yx$-walk $Q=y_1\tp y_2\tp\cdots\tp y_k$ where $P$ and $Q$
avoid each other and all vertices on $P$ and $Q$ are in $Z$.  If both $P$ and
$Q$ avoid $z$, then $\{x,y\}$ is a $z$-anchored invertible pair.  But $G$
contains no such a pair.

So by symmetry between $P$ and $Q$, we may assume that for some $i<k$, the edge
$y_iy_{i+1}$ of $Q$ does not avoid $z$.  Note that $x_i$ avoids $y_iy_{i+1}$,
since $P$ and $Q$ avoid each other.  Thus by \ref{clm:h4}, we conclude that
$y_i,y_{i+1}\in Y$, that $y_i$ overlaps $y_{i+1}$, and that $x_iz\not\in E$.
Since
$y_i,y_{i+1}\in Y$, we deduce by \ref{clm:h0} that $zy_i$ and $zy_{i+1}$ are
1-overlap edges.  So $\bar zy_i$ and $\bar zy_{i+1}$ are also 1-overlap edges,
by~\ref{enum:circ1-ii}.  Similarly, since $x_iz\not\in E$, we have $N[\bar
z]\supseteq N[x_i]$ by \ref{enum:circ1-0}.  This shows that both $y_i$ and
$y_{i+1}$ avoid the edge~$\bar zx_i$.

Finally, recall that $x=x_1=y_k$ and $y=y_1=x_k$, and that $y_{i+1}$ avoids
$x_ix_{i+1}$, since $P$ and $Q$ avoid each other. Thus we define \smallskip

\centereq{
$\begin{array}{@{~}c@{~}c@{~}c@{~}c@{~}c@{~}c@{~}c@{~}c@{~}c@{~}c@{~}c%
@{~}c@{~}c@{~}c@{~}c@{~}c@{~}c@{~}c@{~}c@{~}c@{~}c@{~}c@{~}c@{~}c@{~}c@{~}}
P' & = & y_{i+1} & \tp &  y_{i+1} & \tp &  y_{i+1} & \tp &  y_{i+2} & \tp &
\cdots & \tp &  y_{k-1} & \tp &  x_1 & \tp &  x_2 & \tp & \cdots & \tp &  x_i &
\tp & \bar z\\
Q' & = & \bar z & \tp &  x_i & \tp &  x_{i+1} & \tp &  x_{i+2} & \tp & \cdots &
\tp &  x_{k-1} & \tp &  y_1 & \tp &  y_2 & \tp & \cdots & \tp &  y_i & \tp & y_i
\end{array}$
}\smallskip

We observe that $P'$ and $Q'$ are walks with all vertices in $Z$. By
construction, it follows that $P'$ and $Q'$ avoid each other, where $P'$ is a
$y_{i+1}\bar z$-walk, and $Q'$ is a $\bar zy_i$ walk. Moreover, $y_i,y_{i+1}\in
Y$. This, however, contradicts \ref{clm:h5}. We must therefore conclude that
walks $P$, $Q$ do not exist, which shows \ref{enum:m2}.

Thus $Z$ is indeed a non-inverting set of $G$.
\end{proofclaim}

Now, from \ref{clm:h4} and \ref{clm:h6}, we conclude that $G$ and set $Z$
satisfy the condition of Lemma \ref{lem:non-inverting}. Therefore $G$ is a
circular-arc graph. This proves \ref{enum:I}.

The proof of Theorem \ref{thm:main} is now complete.\hfill\qed
\setcounter{theorem}{\thesave}

\section{Certifying algorithm}\label{sec:certif}

\subsection{Knotting graph}\label{sec:knot}

The notion of \uline{knotting graph} comes from the work of Gallai
\cite{gallai}. It refers to the graph constructed from $G$ as follows: every
vertex $u$ of $G$ is replaced by multiple copies, one copy for each
anticomponent of the neighbourhood of $u$, and
corresponding to each edge $uv$ of $G$, we put an edge between the anticomponent
of $u$ containing $v$ and the anticomponent of $v$ containing $u$.
%each copy is made adjacent to
%every vertex in its corresponding anticomponent.
(Here, an ``anticomponent'' of $N(u)$ refers to a component in the
complement of $G[N(u)]$.) 
This simulates the forcing of
orientations of edges in a construction of a transitive orientation. The
resulting fundamental theorem \cite{gallai} then says that $G$ has a transitive
orientation if and only if the knotting graph of $G$ is bipartite.

We explore an analogous idea for anchored invertible pairs. The corresponding
forcing relation is the $\Delta$-relation introduced in
Section~\ref{sec:edgelabelledgraphs} (cf. \cite{mcc03,probe}); whence our choice
of name.
\smallskip

For a vertex $z\in V$, let $A(z)$ denote the set of vertices that are
non-adjacent to or overlap $z$.

\begin{definition}[$uz$-component]\mbox{}
Vertices $x,y$ are \uline{$uz$-connected} if none of $xu, yu, xz, yz, uz$ is an
inclusion edge, and there exists a path between $x$ and $y$ that avoids both $u$
and $z$. Clearly, if $x,y$ are $uz$-connected, then $x,y\in A(z)\cap A(u)$.
Note that by this definition, a vertex $x$ is $uz$-connected to itself if and
only if $x\in A(z)\cap A(u)$.

A maximal set of pairwise $uz$-connected vertices is called a
\uline{$uz$-component}.
\end{definition}

\begin{definition}[anchored $\Delta$-knotting graph]
Let $z\in V$. For each $u\in A(z)$, let $X^{uz}_1$, $X^{uz}_2$,
\ldots,$X^{uz}_{k_u}$ be the $uz$-components~of~$G$.\smallskip

The \uline{$z$-anchored $\Delta$-knotting graph} of $G$ is the graph ${\cal K}^z_G =
(V^z_K,E^z_K)$ defined as follows:\smallskip
\begin{compactitem}
\item For each $u\in A(z)$, the set $V^z_K$ contains $k_u$ copies $u_1$, \ldots,
$u_{k_u}$ of $u$, one copy for each $X^{uz}_i$.
\item For each $uv\in E$ where $u,v\in A(z)$, there is an edge $u_iv_j$ in
$E^z_K$ if $v\in X^{uz}_i$ and $u\in X^{vz}_j$.
\end{compactitem}
\end{definition}

The corresponding Gallai-type theorem is then as follows.

\begin{theorem}\label{thm:knot-z}
$G$ does not contain a $z$-invertible pair if and only if ${\cal K}^z_G$ is bipartite.
\end{theorem}

\begin{proof}
Let $G=(V,E)$ be a graph and ${\cal K}^z_G$ its the $z$-anchored $\Delta$-knotting graph.

Suppose that ${\cal K}^z_G$ has an odd cycle $x_1\tp x_2\tp\cdots\tp x_k\tp x_1$ where
$k$ is odd.  Consider the vertices $x_1,x_2,x_3$. Since they all belong to
$V^z_K$, there are vertices $v,u,w\in V(G)$ and indices $i,j,k$ such that
$x_1=v_j$, $x_2=u_i$, $x_3=w_k$.  Thus, since $x_1x_2$ and $x_2x_3$ are edges in
$E^z_K$, we conclude that $v\in X^{uz}_i$ and $w\in X^{uz}_i$.  This means that
$v,w$ belong to the same $uz$-component of $G$.  Thus by the definition of a
$uz$-component, there exists a $vw$-path that avoids $u$ and $z$.

We can apply the same argument to other consecutive triples on the cycle.
Namely, since the vertices $x_1,\ldots,x_k$ are in $V^z_K$, there exist vertices
$u_1,u_2,\ldots,u_k\in V$ and indices $i_1,i_2,\ldots,i_k$ such that for all
$1\leq j\leq k$, we have $x_j=(u_j)_{i_j}$.  For convenience, we may define
$u_{k+1}=u_1$ and $i_{k+1}=i_1$, and also $u_0=u_k$ and $i_0=i_k$.  Repeating
the above argument for each $1\leq j\leq k$, we find a $u_{j-1}u_{j+1}$-path
$P_j$ that avoids $u_j$ and $z$.  We define 

\centereq{
$\begin{array}{c@{~}c@{~}c@{~}c@{~}c@{~}c@{~}c@{~}c@{~}c@{~}c@{~}c@{~}c@{~}c%
@{~}c@{~}c@{~}c@{~}c@{~}c@{~}c@{~}c@{~}c@{~}c@{~}c@{~}c@{~}c@{~}c@{~}}
P &  = &  u_1 & \tp & u_1 & \tp & \cdots & \tp & u_1 & \tp & ( &  & P_2 &  & ) &
\tp\quad\cdots\quad\tp & u_k & \tp & \cdots & \tp & u_k & \tp &  u_k\\
Q &  = &  u_k & \tp & ( &  & P_1 &  & ) & \tp &  u_2 & \tp & \cdots & \tp & u_2
& \tp\quad\cdots\quad\tp & ( &  & P_k &  & ) & \tp & u_1
\end{array}$}

Therefore $P$ and $Q$ avoid each other and both walks avoid $z$. Also $P$ is an
$x_1x_k$-walk and $Q$ is an $x_kx_1$-walk. Thus $\{x_1,x_k\}$ is a
$z$-invertible pair.
\smallskip

Conversely, suppose that $G$ contains a $z$-invertible pair $\{v,w\}$. Thus
there exists a $vw$-walk $P$ and a $wv$-walk $Q$ that avoid each other and both
walks avoid $z$.
Let $P$ be the walk $v=x_1\tp x_2\tp\cdots\tp x_t=w$ and $Q$ be the walk $w=y_1
\tp y_2\tp\cdots\tp y_t=v$. Observe that we can assume that for each $i$,
exactly one of $x_i=x_{i+1}$ or $y_i=y_{i+1}$ is true. Therefore, by symmetry of
the walks, we can assume that there exists
$\{t_0,t_1,\ldots,t_k\}\subseteq\{1,2,\ldots,t\}$ such that $1=t_0<t_1<\cdots<
t_{k-1}<t_k=t$ and for each $1\leq j\leq k$, if $j$ is odd, then $x_{t_{j-1}}=
x_{t_{j-1}+1}=\cdots=x_{t_j}=:u_j$ and if $j$ is even, then $y_{t_{j-1}}=y_{t_{j
-1}+1}=\cdots=y_{t_j}=:u_j$. It is not difficult to see that for all $1<j<k$,
there exists an $u_{j-1}u_{j+1}$-walk $P_j$ that avoids $u_j$ and $z$.
If $j$ is even,
then $P_j$ is the walk $u_{j-1}=x_{t_{j-1}}\tp x_{t_{j-1}+1}\tp\cdots\tp x_{t_j}
=u_{j+1}$ and if $j$ is odd, then $P_j$ is the walk $u_{j-1}=y_{t_{j-1}}\tp
y_{t_{j-1}+1}\tp\cdots\tp y_{t_j}=u_{j+1}$. This implies that for $1<j<k$,
the vertices $u_{j-1},u_{j+1}$ belong to the same $u_jz$-component.  Thus we may
define $i_j$ for each $1<j<k$ to be the index such that $u_{j-1},u_{j+1}$ belong
to $X^{u_jz}_{i_j}$. Let $i_1,i_k$ be such that $u_2\in X^{u_1z}_{i_1}$ and
$u_{k-1}\in X^{u_kz}_{i_k}$. (Note that $u_2\in A(u_1)$ as $u_2=y_{t-1}$ and
$u_1=x_{t-1}$ and similarly, $u_{k-1}\in A(u_k)$. Also clearly, $u_2,u_{k-1}\in
A(z)$. Therefore $i_1$ and $i_k$ exist.) Now, for $1\leq j<k$, $(u_j)_{i_j}$ is adjacent
to $(u_{j+1})_{i_{j+1}}$, since $u_j$ belongs to $X^{u_{j+1}z}_{i_{j+1}}$ and
$u_{j+1}$ belongs to $X^{u_jz}_{i_j}$.
From this, we conclude that ${\cal K}^z_G$ contains a walk $(u_1)_{i_1}\tp
(u_2)_{i_2}\tp\cdots\tp (u_{k-1})_{i_{k-1}}\tp (u_k)_{i_k}$.

Note that $u_1=v$.
If $k$ is odd, then $u_k=w$. Observe that the walk $u_{k-1}=y_{t_{k-1}}\tp$
$y_{t_{k-1}+1}\tp\cdots\tp y_t=v=u_1$ avoids $u_k$ and that the walk $u_k=w=
y_1\tp y_2\tp\cdots\tp y_{t_1}=u_2$ avoids $u_1$. Thus we have a $u_{k-1}u_1$
walk that avoids $u_k$ and a $u_ku_2$ walk that avoids $u_1$. From this, we can
conclude that $(u_k)_{i_k}$ is adjacent to $(u_1)_{i_1}$. Thus, $(u_1)_{i_1}\tp
(u_2)_{i_2}\tp\cdots\tp (u_k)_{i_k}\tp (u_1)_{i_1}$ is an odd-length closed walk
in ${\cal K}_G^z$.

If $k$ is even, then $u_k=v=u_1$. In this case, the walk $u_{k-2}=y_{t_{k-2}}
\tp y_{t_{k-2}+1}\tp\cdots\tp y_{t_{k-1}}=u_k=u_1$ avoids $u_{k-1}$ and the
walk $u_{k-1}=x_{t_{k-1}}\tp x_{t_{k-1}+1}\tp\cdots\tp x_t\tp y_1\tp y_2\tp
\cdots\tp y_{t_1}=u_2$ avoids $u_k=u_1$ (note that $x_t=y_1=w$). From this, we
can conclude that $(u_{k-1})_{i_{k-1}}$ is adjacent to $(u_1)_{i_1}$ in ${\cal K}_G^z$.
Thus, $(u_1)_{i_1}\tp (u_2)_{i_2}\tp\cdots\tp (u_{k-1})_{i_{k-1}}\tp
(u_1)_{i_1}$ is an odd-length closed walk in ${\cal K}_G^z$.

We therefore conclude that ${\cal K}_G^z$ is not bipartite.
This completes the proof. 
\end{proof}

\begin{corollary}
$G$ is a circular-arc graph if and only if ${\cal K}^z_G$ is bipartite for each $z\in
V(G)$.
\end{corollary}

\subsection{Algorithm}

The algorithm now follows from Theorem \ref{thm:main} and
Theorem~\ref{thm:knot-z}. Given a graph $G$, we first run a linear-time
(non-certifying)
recognition algorithm \cite{kaplan-nussbaum,mcc03}. If it fails, we construct
the circular completion $H$ of $G$ and pick a minimum-degree vertex $z$ in $H$.
We then construct ${\cal K}^z_{H}$ and find an odd cycle in it.  Using the
cycle, we produce avoiding walks that give a $z$-invertible pair in $H$ using
Theorem~\ref{thm:knot-z}.  Correctness of this procedure is a direct consequence
of the respective theorems. A straightforward implementation that has $O(n^4)$
time complexity is discussed below.

As usual, let $n=|V|$ and $m=|E|$.  To construct the circular completion $H$ of
$G$, we need to remove universal vertices, remove true twins, and then determine
the types of edges in $G$. A linear-time algorithm for this is proposed in
\cite{mcc03}.  Unfortunately, the algorithm assumes $G$ is a circular-arc graph
and fails if it is not, without producing a certificate. Here we perform these
steps directly in time $O(nm)$. Note that from this point on, we can find in
$O(1)$ given $u,v\in V$, whether $uv$ is an edge and if so, what type it has.
Observe that the number of edges in the circular completion is $O(n^2)$ and also
that it can be constructed in time $O(n^2)$. Since the types of the edges in the
circular completion are easily determined using the types of the edges of $G$,
we can compute those in time $O(n^2)$.

Now, let $A(z)$ denote the set of vertices $u$ that are non-adjacent to or
overlap $z$.  To construct ${\cal K}^z_H$, we need to find out for each $u,v$,
the value $\gamma(u,v)=i$ indicating that $v$ belongs to the $i$-th
$uz$-component (for some fixed numbering of these components).  To do that, we
need to find for each $u\in A(z)$, all $uz$-components of $H$.  This is
straightforward. We (temporarily) remove every overlap edge $xy$ where $x,y\in
A(u)\cap A(z)$ and where $u$ or $z$ overlaps both $x$ and $y$. The connected
components formed by the remaining edges inside $A(u)\cap A(z)$ are precisely
the $uz$-components.  We arbitrarily number these components
$X^{uz}_1,\ldots,X^{uz}_{k_u}$ and set $\gamma(u,v)=i$ for every $v$ where $v\in
X^u_i$.  The construction takes $O(n^2)$ time for each $u\in A(z)$. Thus
$O(n^3)$
time overall.  Once the values $\gamma(u,v)$ are known, the construction of
${\cal K}^z_H$ proceeds by definition. The graph has $k_u$ copies of each vertex
$u\in A(z)$. For each edge $uv$ with $u,v\in A(z)$, we add edge $u_iv_j$ where
$\gamma(u,v)=i$ and $\gamma(v,u)=j$.  Note that ${\cal K}^z_H$ has
$O(n^2)$ vertices and $O(n^2)$ edges. Thus the construction of ${\cal
K}^z_H$ takes at most $O(n^2)$ time.  We then find an odd cycle in ${\cal
K}^z_H$, also in $O(n^2)$ time.  Let $C=x_1\tp x_2\tp\cdots\tp x_k\tp x_1$ be
this cycle. We then follow the proof of Theorem~\ref{thm:knot-z}. 

Considering $x_1,x_2,x_3$, there are vertices $u,v,w$ and indices $i,j,\ell$
such that $x_1=v_j$, $x_2=u_i$, and $x_3=w_\ell$.  Thus, there exists a
$vw$-walk $P_1$ in the $uz$-component $X^{uz}_i$.  Finding this walk takes
$O(n^2)$ time, using the same procedure that constructed $uz$-components.
Repeating for each triple of consecutive vertices on $C$, we produce walks $P_i$
which when combined give a $uv$-walk $P$ and a $vu$-walk $Q$ that avoid
each other. By construction the walks avoid $z$. Since ${\cal K}^z_H$, and
therefore $C$, has $O(n^2)$ vertices, these walks can be computed in time
$O(n^4)$.

Thus in time $O(n^4)$ we produce the obstruction as claimed.

\footnotesize
\bibliographystyle{acm}

\end{document}